%% file: main.tex
\documentclass{lmcs}
\pdfoutput=1

\usepackage{lastpage}
\lmcsdoi{21}{1}{4}
\lmcsheading{}{\pageref{LastPage}}{}{}%
{May~03,~2023}{Jan.~14,~2025}{}

\keywords{parameterized verification, finite automata,
regular model-checking} 

\usepackage{hyperref}
\usepackage{tikz}
\usepackage{multirow}
\usepackage{graphicx}
\usepackage{nicefrac}
\usepackage{booktabs}
\usepackage[utf8]{inputenc}
\usepackage{tabularx}
\usepackage{amssymb}
\usetikzlibrary{positioning, calc, shapes.geometric, automata, arrows, backgrounds, decorations.pathreplacing, decorations.pathmorphing, matrix, fit}
\input{macros.tex}

\begin{document}
\title{Regular Model Checking Upside-Down:\texorpdfstring{\\}{}An Invariant-Based Approach}
\titlecomment{This paper is an extended version of \cite{ERW22}. In comparison to the proceedings version, this extended version contains complete and detailed proofs, some additional examples, and a more general notion of inductive invariants based on interpreting regular descriptions (see Section~\ref{subsec:representations-interpreters}). This work was partially funded by the European Research
Council (ERC) under the European Union's Horizon 2020 research and innovation
programme under grant agreement No 787367 (PaVeS)}

\author[J.~Esparza]{Javier Esparza\lmcsorcid{0000-0001-9862-4919}}[a]
\author[M.~Raskin]{Michael Raskin\lmcsorcid{0000-0002-6660-5673}}[b]
\author[C.~Welzel-Mohr]{Christoph Welzel-Mohr\lmcsorcid{0000-0001-5583-0640}}[a]

\address{Technical University of Munich, Boltzmannstraße 3, 85748 Garching b. München}
\email{esparza@in.tum.de, welzel@in.tum.de}

\address{LaBRI, University of Bordeaux, CNRS - UMR 5800, F-33405 Talence CEDEX}
\email{mraskin@u-bordeaux.fr}

\begin{abstract}
  Regular model checking is a technique for the verification
  of infinite-state systems whose configurations can be represented as finite
  words over a suitable alphabet. The form we are studying applies to systems whose set of initial
  configurations is regular, and whose transition relation is captured by a
  length-preserving transducer. To verify safety properties, regular model
  checking iteratively computes automata recognizing increasingly larger
  regular sets of reachable configurations, and checks if they contain unsafe
  configurations. Since this procedure often does not terminate, acceleration,
  abstraction, and widening techniques have been developed to compute a regular
  superset of the reachable configurations.

  In this paper, we develop a complementary procedure. Instead of approaching the
  set of reachable configurations from below, we start with the set of all
  configurations and approach it from above. We
  use that the set of reachable configurations is equal to the intersection of
  all inductive invariants of the system. Since this intersection is 
  non-regular in general, we introduce $b$-invariants, defined as those
  representable by CNF-formulas with at most $b$ clauses. We prove that, for
  every $b \geq 0$, the intersection of all inductive $b$-invariants is
  regular, and we construct an automaton recognizing it. We
  show that whether this automaton accepts some unsafe
  configuration is in EXPSPACE  for every $b
  \geq 0$, and PSPACE-complete for $b=1$. Finally, we study how large must $b$ be to prove safety properties of a number of benchmarks.
\end{abstract}

\maketitle

\section{Introduction}
\label{sec:introduction}

Regular model checking (RMC) is a framework for the verification of different classes of infinite-state systems (see, e.g., \ the surveys \cite{AbdullaJNS04,Abdulla12,AbdullaST18,Abdulla21}). In its canonical version, RMC is applied to systems satisfying the following conditions: configurations can be encoded as words,  the set of initial configurations is recognized by a finite automaton $\initial$, and the transition relation is recognized by a length-preserving transducer $\transitions$. 
RMC algorithms address the problem of, given a regular set of unsafe configurations, deciding if its intersection with the set of reachable configurations is empty or not. 
In the present paper, we do not consider generalisations
to non-length-preserving or non-finite-state-transducer transitions.

The fundamental building block of current RMC algorithms is an
automata-theoretic construction that, given a non-deterministic automaton (NFA) $A$ recognizing a regular set of configurations, produces another NFA recognizing the set of immediate successors (or predecessors) of $\lang{A}$ with respect to the transition relation represented by $\transitions$. Therefore, if some unsafe configuration is reachable, one can find a witness by, starting with the automaton $\initial$ for the set of initial configurations, repeatedly adding the set of immediate successors. However, this approach almost never terminates when all reachable configurations are safe. Research on RMC has produced many acceleration, abstraction, and widening techniques to make the iterative computation ``jump over the fixpoint'' in finite time, and produce an invariant of the system not satisfied by any unsafe configuration (see, e.g., \cite{BouajjaniJNT00,JonssonN00,DamsLS01,AbdullaJNd02,BoigelotLW03,BouajjaniHV04,BouajjaniT12,BouajjaniHRV12,Legay12,DBLP:conf/fmcad/ChenHLR17}). 

In this paper, we develop a complementary approach that, starting with the set of all configurations, computes increasingly smaller regular inductive invariants, i.e.,  sets of configurations closed under the reachability relation and containing all initial configurations. Our main contribution is the definition of a sequence of \emph{regular} inductive invariants that converges (in the limit) to the set of reachable configurations, and for which automata can be directly constructed from $\initial$ and $\transitions$. 

While some of the previous work (e.g. using abstraction \cite{BouajjaniHV04}) does include overapproximation, this comes from replacing automata with smaller more permissive ones and then using acceleration techniques. In contrast, our work directly constructs regular-language-described properties that are satisfied by all the reachable configurations of the original system.

Our starting point is the fact that the set of reachable configurations is equal to the intersection of all inductive invariants. Since this intersection is non-regular in general, we introduce \emph{$b$-invariants}. An invariant is $b$-bounded, or just a $b$-invariant, if for every $\ell \geq 0$, the configurations of length $\ell$ satisfying the invariant are those satisfying a Boolean formula in conjunctive normal form with at most $b$ clauses. The atomic propositions are claims of the form ``at the position $i$ there is character $x$''. For example, assume that the configurations of some system are words over the alphabet $\{a,b,c,d\}$, and that the configurations of length five where the second letter is an $a$ or the fourth letter is a $b$, and the second letter is a $b$ or the third is a $c$, constitute an inductive invariant. Then this set of configurations is a $2$-invariant, represented by the formula $(a_{2:5} \vee b_{4:5}) \wedge (b_{2:5} \vee d_{3:5})$.
We prove that, for every bound $b \geq 0$, the intersection of all inductive $b$-invariants, denoted  $\Ind{b}$, is regular, and recognized by a DFA of double exponential size in $\initial$ and $\transitions$. As a corollary, we obtain that, for every $b \geq 0$, deciding if $\Ind{b}$ contains some unsafe configuration is in \textsc{EXPSPACE}. 
Moreover, the proof inspires a wider class of 
regular inductive invariants.
Introducing this class is a central contribution of this paper which is an extended version of \cite{ERW22}.
To define such an invariant, we pick a finite transducer and a regular language of words, the invariant being the image of the language
under the action of the transducer.
Generally, we pick a single transducer to analyse a given regular transition system, and study the invariants provided by different regular languages.
We show that the \textsc{EXPSPACE} upper bound still holds for the wider class.

In the second part of the paper, we study the special case $b=1$ in more
detail. We exploit that inductive $1$-invariants are closed under union
(a special feature of the  $b = 1$ case), and prove that deciding whether $\Ind{1}$
contains some unsafe configuration is \textsc{PSPACE}-complete. The proof also
shows that $\Ind{1}$ can be recognized by an NFA of single exponential size in $\initial$ and $\transitions$.

The index $b$ of a bounded invariant can be seen as a measure of how difficult it is for a human to understand it. So one is interested in the smallest $b$ such that $\Ind{b}$ is strong enough to prove a given property.  In the third and final part of the paper, we experimentally show that for a large number of systems $\Ind{1}$ is strong enough to prove useful safety properties.

\subsubsection*{Related work.} The work closest to ours is \cite{AbdullaDHR07}, which directly computes an overapproximation of the set of reachable configurations of a parameterized system. Contrary to our approach, the paper computes one single approximation, instead of a converging sequence of overapproximations. Further, the method is designed for a model of parameterized systems with existential or universal guarded commands, while our technique can be applied to any model analyzable by RMC.  Our work is also related to \cite{DBLP:conf/fmcad/ChenHLR17}, which computes an overapproximation using a learning approach, which terminates if the set of reachable configurations is regular; our paper shows that a natural class of invariants is regular, and that automata for them can be constructed explicitly from the syntactic description of the system. This paper generalizes the work of \cite{DBLP:conf/tacas/BozgaEISW20,DBLP:journals/corr/abs-2108-09101,DBLP:journals/jlap/BozgaIS21,DBLP:conf/apn/EsparzaRW21} on trap invariants for parameterized Petri nets. Trap invariants are a special class of $1$-bounded invariants, and the parameterized Petri nets studied in these papers can be modelled in the RMC framework. An alternative to regular model checking are logical based approaches. The invisible invariant method synthesizes candidate invariants from examples, which are then checked for inductiveness \cite{PnueliRZ01}. Our approach does not produce candidates, it generates invariants by construction.  Modern tools like Ivy \cite{PadonMPSS16,McMillanP20} have verified more complex protocols than the ones in Section~\ref{sec:experiments} using  a combination of automation and human interaction. The best way of achieving this interaction is beyond the scope of this paper, which focuses on the foundations of regular model checking.

\subsubsection*{Structure of the paper.} Section~\ref{sec:prelims} introduces basic definitions of the RMC framework. Section~\ref{sec:boundedinv} introduces $b$-bound invariants while Section~\ref{sec:indinv-is-regular} proves regularity of $\Ind{b}$. Sections~\ref{sec:pspacecomp} and~\ref{sec:pspacehardness} prove the
\textsc{PSPACE}-completeness result. Sections~\ref{sec:experiments} and~\ref{sec:conclusions} contain some experimental results and conclusions.

\section{Preliminaries}
\label{sec:prelims}

Given $n, m \in \mathbb{N}$, we let $[n,m]$ denote the  set $\{ i \in \mathbb{N} \colon n \leq i \leq m\}$.

\paragraph{Languages and automata.}
A language over a finite alphabet $\Sigma$ is a subset of $\Sigma^*$. 
An element of a language is called a word; for a word $w$ we denote $w[i]$ its $i$-th letter.
Given a language $L \subseteq \Sigma^*$,
we let $\comp{L}$ denote the language $\Sigma^* \setminus L$. A nondeterministic finite automaton (NFA) is a tuple $A = \tuple{Q, \Sigma, \Delta, Q_{0}, F}$ where $Q$
is a non-empty finite set of states, $\Sigma$ is an alphabet, $\Delta \colon Q \times \Sigma \rightarrow 2^Q$
is a transition function, and $Q_0, F \subseteq Q$ are sets of initial and final states, respectively.
A run of ${A}$ on a word $w \in \Sigma^{\ell}$ is a sequence $q_{0} \, q_{1} \,
\ldots \, q_{\ell}$ of states such that $q_{0} \in Q_{0}$ and $q_{i} \in \Delta(q_{i-1},w[i])$ for every $i \in [1, \ell]$. A run on $w$ is accepting if $q_\ell \in F$, and ${A}$ accepts $w$ if there exists an accepting run of ${A}$ on $w$.
The language recognized by ${A}$, denoted $\lang{{A}}$ or
$\varlang{{A}}$, is the set of words accepted by ${A}$. We
let $\size{{A}}$ denote the number of states of  ${A}$.
The function $\delta_{{A}} \colon 2^{Q} \times
\Sigma^{*} \rightarrow 2^Q$ is defined inductively as follows: $\delta_{{A}}(P, \varepsilon) = P$ and $\delta_{{A}}(P, a w) = \delta_{{A}}(\bigcup_{p \in P} \Delta(p, a), w)$.
Observe that ${A}$ accepts $w$ if{}f $\delta_{{A}}(Q_0, w) \cap F \neq \emptyset$. An NFA ${A}$ is deterministic\footnote{Sometimes in the literature this is called ``deterministic and complete'' with determinism being a weaker property. However, unlike the size gap between NFAs and DFAs, the difference is limited to the presence of a single ``useless'' bottom state.} (\emph{DFA}) if $\size{Q_{0}} = 1$ and $|\Delta(q, a)| = 1$ for every $q \in Q$ and $a \in \Sigma$.

\paragraph{Length-preserving relations and length-preserving transducers.}  A \emph{length-preserving relation}
over an alphabet $\Sigma \times \Gamma$ is a language over the alphabet $\Sigma \times \Gamma$.
We denote elements of $\Sigma \times \Gamma$ as $\tuple{a,b}$ or 
$\vtuple{a \\ b}$ where $a \in \Sigma$ and $b \in \Gamma$. Given $w=w[1] \ldots w[\ell] \in \Sigma^*$ and
$u=u[1] \ldots u[\ell] \in \Gamma^*$, we let $\tuple{w,u}$ denote the word $\tuple{w[1], u[1]} \cdots \tuple{w[\ell], u[\ell]} \in (\Sigma \times \Gamma)^*$.
That is, if we write $\tuple{w,u}$ then necessarily $w$ and $u$ have the same length. We look at $\tuple{w,u}$ as a representation of the
pair of words $(w, u)$.

\begin{rem}
Throughout the paper, we consider only length-preserving relations on words, and so we call them just relations.
\end{rem}

The \emph{complement} of a relation $R \subseteq (\Sigma \times \Gamma)^*$ is the relation
$\comp{R}:= \{ \tuple{w,u} \in  (\Sigma \times \Gamma)^*  \mid \tuple{w,u} \notin R\}$. Sometimes we represent relations 
by infix operators, like $\leadsto$, and then we write $\compl{\leadsto}$ to denote the complement. Observe that, by definition,
the complement of a relation is a subset of $(\Sigma \times \Gamma)^*$, and so it only contains tuple $\tuple{w, u}$ where $|w|=|u|$.
The \emph{join} of two relations $R_1 \subseteq (\Sigma_1 \times \Gamma)^*$ and $R_2 \subseteq (\Gamma \times \Sigma_2)^*$,
is the relation $R_1 \circ R_2 \subseteq (\Sigma_1 \times \Sigma_2)^*$ given by: $\tuple{w_1, w_2}  \in R_1 \circ R_2$ if{}
there exists $w \in \Gamma^*$ such that $ \tuple{w_1, w} \in R_1$ and $\tuple{w, w_2} \in R_2$. The \emph{post-image} of a language $L \subseteq \Sigma^*$
under a relation $R \subseteq (\Sigma \times \Gamma)^*$ is the language $L \circ R$ given by: $w \in L \circ R$ if{}f there exists $u \in \Sigma^*$
such that $u \in L$ and $\tuple{u, w} \in R$. The \emph{pre-image} of $L$ under $R$, denoted $R \circ L$, is defined analogously. The 
\emph{projections} of a relation $R \subseteq (\Sigma \times \Gamma)^*$ onto its first and second components are the languages
$\proj{R}{1} := \{ w \in \Sigma^* \mid \exists u \in \Gamma^* : \tuple{w,u} \in R \}$ and $\proj{R}{2} := \{ u \in \Gamma^* \mid \exists w \in \Sigma^* : \tuple{w,u} \in R \}$.
The \emph{inverse} of a relation $R \subseteq (\Sigma \times \Gamma)^*$ is the relation $R^{-1} := \{ \tuple{u, w} \in (\Gamma \times \Sigma)^* \mid \tuple{w, u} \in R \}$.

A (length-preserving) \emph{transducer} over $\Sigma \times \Gamma$ is an NFA with $\Sigma \times \Gamma$ as alphabet. 
The (length-preserving) relation recognized by a transducer $T$, denoted $\rel{T}$ or $\varrel{T}$,  is the set of tuples $\tuple{w, u} \in (\Sigma \times \Gamma)^*$ accepted by $T$.
A relation is \emph{regular} if it is recognized by a transducer. A transducer is \emph{deterministic} if it is a deterministic NFA.

It is folklore that regular relations are closed under Boolean operations (the operations are implemented as for regular languages) and composition, also called join, and that
the pre- and post-images of regular languages under regular relations are regular. We sketch the constructions in the following proposition.

\begin{prop}\label{prop:transducerops}
Let $\Sigma, \Gamma$ be finite alphabets.
\begin{enumerate}
\item Let $T, U$ be transducers with $n_T$ and $n_U$ states over alphabets $\Sigma_T \times \Gamma$ and $\Gamma \times \Sigma_U$, 
respectively. There  exists a transducer with $\O(n_T n_U)$ states recognizing $\rel{T} \circ \rel{U}$.
\item Let $A, B$ be NFAs over $\Sigma$, respectively $\Gamma$, with $n_A$ and $ n_B$ states, respectively,  and let $T$ be a transducer over $\Sigma \times \Gamma$ with $n_T$ states. There exist NFAs with $\O(n_A \cdot n_T)$ and $\O(n_B \cdot n_T)$ states recognizing $\lang{A} \circ \rel{T}$ and $\rel{T} \circ \lang{B}$, respectively. 
\item Let $T$ be a transducer with $n_T$ states. There exist NFAs  with $n_T$ states recognizing
$\proj{\rel{T}}{1}$ and  $\proj{\rel{T}}{2}$.
\item Let $T$ be a transducer with $n_T$ states. There exists a transducer with $n_T$ states recognizing
$\rel{T}^{-1}$.
\end{enumerate}
\end{prop}
\begin{proof}
1) Let $T = \tuple{Q_T, \Sigma_T \times \Gamma,  \Delta_T, Q_{0T}, F_T}$ and $T_U = \tuple{Q_U, \Gamma \times \Sigma_U,  \Delta_U, Q_{0U}, F_U}$. Define
the transducer 
$V:= \tuple{Q_T \times Q_U, \Sigma_T \times \Sigma_U, \Delta, Q_{0T} \times Q_{0U}, F_T \times F_U}$. For this transducer, we fix $(\tuple{q_T, q_U}, \tuple{\sigma_T, \sigma_U}, \tuple{q_T', q_U'} ) \in \Delta$ 
if{}f there exists $\gamma \in \Gamma$ such that $(q_T,  \tuple{\sigma_T, \gamma}, q_T') \in \Delta_T$ and $(q_U,  \tuple{\gamma, \sigma_U}, q_U') \in \Delta_U$.
We have $\rel{V} = \rel{T} \circ \rel{U}$.

\medskip\noindent 2) Let $A = \tuple{Q_A, \Sigma, \Delta_A, Q_{0A}, F_A}$ and $T = \tuple{Q_T, \Sigma \times \Gamma,  \Delta_T, Q_{0T}, F_T}$. Define 
the NFA $A' = \tuple{Q_A \times Q_T, \Gamma, \Delta, Q_{0A} \times Q_{0T}, F_A \times F_T}$ where $(\tuple{q_A, q_T}, \gamma, \tuple{q_A', q_T'} ) \in \Delta$ 
if{}f there exists $\sigma \in \Sigma$ such that $(q_A,  \sigma, q_A') \in \Delta_A$ and $(q_T,  \tuple{\sigma, \gamma}, q_T') \in \Delta_T$.
We have $\lang{A'}= \lang{A} \circ \rel{T}$. The construction for $\rel{T} \circ \lang{B}$ is analogous.

\medskip\noindent 3) Let $T = \tuple{Q_T, \Sigma \times \Gamma,  \Delta_T, Q_{0T}, F_T}$. Define $A_1 = \tuple{Q_T, \Sigma,  \Delta, Q_{0T}, F_T}$
where $(q, \sigma, q') \in \Delta$ if{}f there exists $\gamma \in \Gamma$ such that $(q, \tuple{\sigma, \gamma}, q') \in \Delta_T$. We have $\lang{A_1}= \proj{R}{1}$.
The NFA $A_2$ is defined analogously.

\medskip\noindent 4) Let $T = \tuple{Q_T, \Sigma \times \Gamma,  \Delta_T, Q_{0T}, F_T}$. Define $T^{-1} =\tuple{Q_T, \Gamma \times \Sigma,  \Delta', Q_{0T}, F_T}$
where $\Delta'(q, \allowbreak\tuple{a,b}) := \Delta_T(q, \tuple{b, a})$ for all $q, a, b$.
\end{proof}

\paragraph{Regular model checking.}
Regular model checking (RMC) is a framework for the verification of systems with infinitely many
configurations. Each configuration is represented as a finite word over a fixed
alphabet $\Sigma$. Systems are modelled as regular transition systems of the following form.

\begin{defi}[Regular transition systems]
  A \emph{regular transition system} (RTS) is a triple $\mathcal{R} = \tuple{\Sigma, \initial,
  \transitions}$ where $\Sigma$ is an alphabet, 
 $\initial$ is an NFA over $\Sigma$, and $\transitions$ is a transducer over $\Sigma \times \Sigma$.
\end{defi}

Words over $\Sigma$ are called \emph{configurations}. Configurations accepted by 
$\initial$ are called \emph{initial}, and pairs of configurations accepted by $\transitions$ are called \emph{transitions}.  We write $w \leadsto u$ to denote that $\tuple{w, u}$ is a transition. Observe that $w \leadsto u$ implies $|w|=|u|$. Given two configurations $w, u$, we say that $u$ is \emph{reachable} from $w$ if $w \reaches u$ where $\reaches$ denotes the reflexive and transitive closure of $\leadsto$. The set of reachable configurations of $\mathcal{R}$, denoted $\reach(\mathcal{R})$, or just $\reachconf{}$ when there is no confusion, is the set of configurations reachable from the initial configurations.  In the following, we use $\size{\mathcal{R}}$ to refer to $\size{\initial} + \size{\transitions}$.

\begin{exa}[Dining philosophers]
\label{ex:phil}
We model a very simple version of the dining philosophers as an RTS, for use as running example. Philosophers sit at a round table with forks between them. Philosophers can be \emph{thinking} ($\tilo$) or \emph{eating} ($\eilo$). Forks can be \emph{free} ($\fork$) or \emph{busy} ($\bork$). A thinking philosopher whose left and right forks are free can \emph{simultaneously} grab both forks---the forks become busy---and start eating. 
        After eating, the philosopher puts both forks to the table and returns to thinking. The model includes two corner cases: a table with one philosopher and one fork, which is then both the left and the right fork (unusable as it would need to be grabbed twice in a single transition), and the empty table with no philosophers or forks. 

We model the system as an RTS $\rts = \tuple{\Sigma, \initial, \transitions}$ over the alphabet $\Sigma = \set{\tilo, \eilo, \fork,
\bork}$. A configuration of a table with $n$ philosophers and $n$ forks is represented
as a word over $\Sigma$ of length $2n$. Letters at odd and even positions 
model the current states of philosophers and forks (positions start at $1$). For example, $\tilo\fork\tilo\fork$ models a table with two thinking philosophers and two free forks. The set of initial configurations is $\lang{\initial}=(\tilo\fork)^*$, and the set of transitions is
\[ \rel{\transitions} = \vtuple{t \\e}\vtuple{\fork \\ \bork}\vtuple{x \\x}^* \vtuple{\fork \\ \bork} \;\;\bigg|\;\; \vtuple{e \\t}\vtuple{\bork \\ \fork}\vtuple{x \\x}^* \vtuple{\bork \\ \fork} \;\;\bigg|\;\; \vtuple{x \\x}^*\left( \vtuple{\fork \\ \bork}\vtuple{\tilo \\ \eilo}\vtuple{\fork \\ \bork} \;\bigg|\; \vtuple{\bork \\ \fork}\vtuple{\eilo \\ \tilo}\vtuple{\bork \\ \fork} \right) \vtuple{x \\x}^*  \]
\noindent  where $\vtuple{x \\x}$ stands for the regular expression $\left( \vtuple{\tilo \\ \tilo} \;\bigg|\; \vtuple{\eilo \\ \eilo} \;\bigg|\; \vtuple{\fork \\ \fork} \;\bigg|\; \vtuple{\bork \\ \bork}\right)$. The first two terms of $\rel{\transitions}$ describe the actions of the first philosopher, and the second the actions of the others. 
It is not difficult to show that $\reachconf{} = (\tilo ( \fork \,|\, \bork\eilo\bork))^* \;|\; \eilo\bork\tilo (( \fork \,|\, \bork\eilo\bork) \tilo)^*\bork$. These are the configurations where no two philosophers are using the same fork, and fork states match their adjacent philosopher states.
\end{exa}

\paragraph{Safety verification problem for RTSs.}
The safety verification problem for RTSs is defined as follows: Given an RTS
$\mathcal{R}$ and an NFA $\mathcal{U}$ recognizing a set of \emph{unsafe} configurations, decide whether $\reach(\mathcal{R}) \cap \lang{\mathcal{U}} = \emptyset$ holds.  The problem is known to be undecidable. 

\begin{exa}\label{ex:dead}
A configuration $w$ of an RTS is \emph{deadlocked} if there is no configuration $u$ such that $w \leadsto u$. 
It is easy to see that the set of deadlocked configurations of the dining philosophers of \autoref{ex:phil} is
        \[ \mathit{Dead}=\comp{\Sigma^* \fork \, \tilo \, \fork \, \Sigma^*} \cap \comp{\Sigma^* \bork \, \eilo \, \bork  \, \Sigma^*} \cap 
        \comp{\tilo \, \fork \, \Sigma^* \fork } \cap \comp{\eilo \, \bork  \, \Sigma^* \bork }  \ . \] 
In other words, these are configurations containing neither $\fork \, \tilo \, \fork$
not $\bork \, \eilo \, \bork$ as a cyclic word.
The dining philosophers are deadlock-free if{}f $\reachconf{} \cap \mathit{Dead} = \emptyset$, which is the case (recall that philosophers can only grab both forks simultaneously).
\end{exa}

\section{Bounded inductive sets of an RTS}
\label{sec:boundedinv}

We present an invariant-based approach to the safety verification problem for
RTSs. Fix  an RTS $\mathcal{R} = \tuple{\Sigma, \initial, \transitions}$. We introduce an infinite sequence 
\[ \Sigma^* = \Ind{0} \supseteq  \Ind{1} \supseteq \Ind{2} \ldots \supseteq \reachconf{} \]
\noindent  of effectively regular inductive invariants of $\mathcal{R}$ that converges to $\reachconf{}$, i.e., $\Ind{k}$ is effectively regular for every $k\geq 1$, and $\reachconf{} = \bigcap_{k=0}^\infty \Ind{k}$. Section~\ref{subsec:basic} recalls basic notions about inductive sets and invariants, and Section~\ref{subsec:bounded} defines the inductive invariant $\Ind{b}$ for every $b \ge 0$.

\subsection{Inductive sets and invariants}
\label{subsec:basic}

Let $S \subseteq \Sigma^*$ be a set of configurations. $S$ is \emph{inductive} if it is closed under reachability,  i.e., if $w\in S$ and $w\leadsto u$ implies $u\in S$. 
Observe that inductive sets are closed under union and intersection. $S$ is an  \emph{invariant for length $\ell$} if $\reachconf{} \cap \Sigma^\ell \subseteq S \cap \Sigma^\ell$, and an \emph{invariant} if $\reachconf{} \subseteq S$. Observe that, since $\transitions$ is a length-preserving transducer, $S$ is an invariant if{}f it is an invariant for every length. 
Given two invariants $\invar{}_1, \invar{}_2$, we say that $\invar{}_1$ is \emph{stronger} than  $\invar{}_2$ if $\invar{}_1 \subset \invar{}_2$. 

A set $S$ is an \emph{inductive invariant} if it is both inductive and an invariant, i.e., if it contains $\reachconf{}$ and is closed under reachability. 
Every inductive invariant $S$ 
and every initial configuration $I$
satisfy $I\in{}S$. There is a unique smallest inductive invariant w.r.t. set inclusion, namely the set $\reachconf{}$ itself. 

\begin{exa}
\label{ex:reachphil}
The set $I_0 = ((\tilo \,|\, \eilo)(\fork \,|\, \bork))^*$ is an inductive invariant of the dining philosophers. Other inductive invariants are
        \[  
              I_1= \comp{\Sigma^* \eilo \fork \eilo \Sigma^*}, \quad
              I_2= \comp{\eilo \, \Sigma^* \eilo \fork}, \quad
              I_3 = \comp{\Sigma^* \tilo \, \bork \, \tilo  \Sigma^*}, \quad 
              I_4 = \comp{\tilo \, \Sigma^* \tilo \, \bork }, \quad 
              \]\[
              I_5 = \comp{(\tilo \, \Sigma \, \eilo \, \Sigma)^*} , \quad
              I_6 = \comp{(\eilo \, \Sigma \, \tilo \, \Sigma)^*} .
        \]
\noindent  Taking into account that the table is round, these are the sets of configurations 
        without any occurrence of $\eilo\fork\eilo$ ($I_1$ and $I_2$) 
        and $\tilo \, \bork \, \tilo$ ($I_3$ and $I_4$) as a cyclic word; 
        as well as without alternation of $\tilo$ and $\eilo$ throughout the \emph{entire} 
        configuration ($I_5$ and $I_6$).
        Note that the latter condition is vacuously true for an odd number of philosophers.
\end{exa}

\subsection{Bounded inductive sets.}
\label{subsec:bounded}
Given a length $\ell \geq 0$, we represent certain sets of configurations as Boolean formulas over a set $\mathit{AP}_\ell$ of atomic propositions. More precisely, a Boolean formula over $\mathit{AP}_\ell$ describes a set containing \emph{some} configurations of length $\ell$,  and \emph{all} configurations of other lengths.

The set $\mathit{AP}_\ell$ contains an atomic proposition $\vari{q}{j}{\ell}$
for every  $q \in \Sigma$ and for $j \in [1, \ell]$.  A \emph{formula}
$\varphi$ over $\mathit{AP}_\ell$ is a positive Boolean combination of atomic
propositions of $\mathit{AP}_\ell$ and the constants $\true$ and $\false$.
Formulas are interpreted on configurations. Intuitively, an atomic proposition
$\vari{q}{j}{\ell}$ states that either the configuration does \emph{not} have
length $\ell$, or it has length $\ell$ and its $j$-th letter is $q$. Formally,
$w \in \Sigma^*$ \emph{satisfies} the single atomic proposition $\vari{q}{j}{\ell}$
if either $|w| \neq \ell$ or $|w|=\ell$ and $w[j] = q$.
For non-atomic formulas, i.e., $\varphi = \true, \varphi_1 \vee
\varphi_2, \varphi_1 \wedge \varphi_2$, satisfaction is defined as usual.  The
language $\lang{\varphi} \subseteq \Sigma^*$ of a formula is the set of
configurations that satisfy $\varphi$. We also say that $\varphi$
\emph{denotes} the set $\lang{\varphi}$.
A formula is inductive if it denotes an inductive set.

\begin{exa}
  In the dining philosophers, let $\varphi = (\vari{\eilo}{1}{4}
  \wedge\vari{\bork}{4}{4}) \vee \vari{\fork}{2}{4}$. We have
        \[ 
    \lang{\varphi} = \epsilon \mid \Sigma \mid \Sigma^2 \mid \Sigma^3
  \mid 
    \eilo \; \Sigma \; \Sigma \; \bork \mid \Sigma \; f \; \Sigma \; \Sigma
  \mid \Sigma^5\Sigma^* \ .
        \] 
\end{exa}

Observe that an expression like $(\vari{q}{1}{1} \wedge \vari{r}{1}{2})$ is not a formula because it combines atomic propositions of two different lengths, which is not allowed. 
Notice also that $\neg \vari{q}{j}{\ell}$ is equivalent to $\bigvee_{r \in \Sigma \setminus \{q\}} \vari{r}{j}{\ell}$. 
Therefore, if we allowed negative atomic propositions, we would still have the same class
of expressible predicates on words of a given length (and we would not obtain formulas for the same predicates with fewer clauses.)
Abusing language, if $\varphi$ is a formula over $\mathit{AP}_\ell$ and $\lang{\varphi}$ is an (inductive) invariant, then we also say that $\varphi$ is an (inductive) invariant. Observe that (inductive) invariants are closed under conjunction and disjunction.

\begin{center}
\underline{Convention}: From now on, ``formula'' means ``positive formula in CNF''.
\end{center}

\begin{defi}
Let $b \geq 0$. A \emph{$b$-formula} is a formula with at most $b$ clauses (with the convention that $\true$ is the only formula with $0$ clauses). A set $S \subseteq \Sigma^*$ of configurations is \emph{$b$-bounded}  if for every length $\ell$ there exists a $b$-formula $\varphi_\ell$ over $\mathit{AP}_\ell$ such that $S \cap \Sigma^\ell = \lang{\varphi_\ell}$.  We abbreviate $b$-bounded sets to just $b$-sets. We call a $b$-bounded invariant a $b$-invariant.
\end{defi}

Observe that, since one can always add tautological clauses to a formula without changing its language, a set $S$ is $b$-bounded if{}f for every length $\ell$ there is a formula $\varphi_\ell$ with \emph{exactly} $b$ clauses.

\begin{exa}
\label{ex:binv}
  In the dining philosophers,  the $1$-formulas $(\vari{\tilo}{2i-1}{\ell} \vee
  \vari{\eilo}{2i-1}{\ell})$ and $(\vari{\fork}{2i}{\ell} \vee
  \vari{\bork}{2i}{\ell})$ are inductive $1$-invariants for every even $\ell \geq 1$ and every $i \in
  [1, \ell/2]$. It follows that the
  set $I_0$ of \autoref{ex:reachphil} is an intersection of (infinitely
  many) such inductive $1$-invariants. The same happens for $I_1, \ldots, I_6$. 
  For example, 
        $I_1$ is the intersection of all inductive $1$-invariants of the
  form $(
        \vari{\tilo}{i}{\ell} 
        \vee \vari{\fork}{i}{\ell} 
        \vee \vari{\bork}{i}{\ell} 
        \vee 
        \vari{\bork}{i+1}{\ell}
        \vee \vari{\tilo}{i+1}{\ell}
        \vee \vari{\eilo}{i+1}{\ell}
        \vee
        \vari{\tilo}{i+2}{\ell} 
        \vee \vari{\fork}{i+2}{\ell} 
        \vee \vari{\bork}{i+2}{\ell} 
        )$, for all $\ell
  \geq 1$ and all $i \in [1, \ell-1]$; inductivity is shown by an easy case
  distinction.
\end{exa}

We introduce the inductive invariants studied in the paper.

\begin{defi}\label{def:indb}
Let $\rts=(\Sigma, \initial, \transitions)$ be an RTS, let $w, u \in \Sigma^*$ be two configurations of $\rts$, and let $b \geq 0$. We say that $w$ is \emph{$b$-potentially-reachable} from $u$, denoted $u \reachesb{b} w$, if every inductive $b$-formula satisfied by $u$ is also satisfied by $w$. In other words, no inductive $b$-set provides an explanation why $w$ should not be reachable from $u$. Further, we define $\Ind{b} = \{ w \in \Sigma^* \mid  u \reachesb{b} w \text{ for some initial configuration } u  \}$.
\end{defi}

Observe that $u \reaches w$ implies $u \reachesb{b}w$ for every $b \geq 1$. Indeed, if $u \reaches w$ then, by the definition of an inductive formula, every inductive formula satisfied by $u$ is also satisfied by $w$.  Therefore, if $w$ is reachable from $u$ then it is also $b$-potentially  reachable from $u$ for every $b \geq 0$. However, the converse does not necessarily hold.

The following proposition shows that  $\Ind{b}$ is an inductive invariant for every $b \geq 0$, and some fundamental properties.

\begin{prop}\label{prop:inter}
  Let $\mathcal{R}$ be an RTS. For every $b \geq 0$: 
  \begin{enumerate}
  \item $\Ind{b}$ is an inductive invariant, and $\Ind{b} \subseteq S$ for every inductive $b$-invariant $S$.
  \item $\Ind{b} \supseteq \Ind{b+1}$. 
  \item $\reachconf{} = \bigcap_{b=0}^\infty \Ind{b}$.
  \end{enumerate}
\end{prop}
\begin{proof}
\noindent (1)  We first show that $\Ind{b}$ is an inductive invariant. 
\begin{itemize}
\item $\Ind{b}$ is inductive. Let $w \in\Ind{b}$ and $w \leadsto v$. Since $w \in\Ind{b}$, we have $u \reaches_b w$ for
some initial configuration $u$. We prove $u \reaches_b v$. Let $\varphi$ be an inductive $b$-formula satisfied by $u$. By definition, $\varphi$ is also satisfied by $w$ and, since $\varphi$ is inductive
and $w \leadsto v$, also by $v$. 
\item $\Ind{b} \supseteq \reachconf{}$. Since $\Ind{b}$ is inductive, it suffices to show that it contains all initial configurations. This follows immediately from the definition of $\Ind{b}$ and the fact that every configuration is $b$-potentially reachable from itself.
\end{itemize}
For the second part, let $S$ be an arbitrary $b$-invariant, and let $w \in \Ind{b} $. We prove $w \in S$. Let $\ell$ be the length of $w$. Since $S$ is a $b$-invariant, there exists a $b$-formula
$\varphi$ over $AP_\ell$ such that $S \cap \Sigma^\ell = \lang{\varphi}$. So it suffices to prove $w \models \varphi$. 
Since $w \in \Ind{b}$, there exists an initial configuration $u$ such that every inductive $b$-formula satisfied by $u$ is also satisfied by $w$. In particular, this also holds for $\varphi$.
Further, since $S$ is an invariant, $u \models \varphi$. So $w \models \varphi$, and we are done.

\medskip\noindent (2) $\Ind{b} \supseteq \Ind{b+1}$ follows from the fact that, by definition, every $b$-formula is also a $(b+1)$-formula. 

\medskip\noindent (3)  For every $\ell \geq 0$, the set $\reachconf{} \cap \Sigma^\ell$ is an inductive invariant for length $\ell$. Let $\varphi_\ell$ be a formula over $\mathit{AP}_\ell$ such that $\lang{\varphi_\ell} \cap \Sigma^\ell = \reachconf{} \cap \Sigma^\ell$, and let $b_\ell$ be its number of clauses. (Notice that $\varphi_\ell$ always exists because
every subset of $\Sigma^\ell$ can be expressed as a formula, and every formula can be put in conjunctive normal form.)
Then $\varphi_\ell$ is a $b_\ell$-bounded invariant, and so
  $\lang{\varphi_\ell}  \supseteq \Ind{b_\ell}$ for every $\ell \geq 0$.
  So we have $\reachconf{} = \bigcap_{\ell=0}^\infty \lang{\varphi_\ell}
  \supseteq \bigcap_{b=0}^\infty \Ind{b}$.
        As each $\Ind{b}$ is an invariant, $\Ind{b}\supseteq\reachconf{}$
        for all $b$, thus $\bigcap_{b=0}^\infty \Ind{b}\supseteq\reachconf{}$
        and we are done.
\end{proof}
\noindent 
By (1), the invariant $\Ind{b}$ is as strong as any  $b$-invariant. Notice, however, that  $\Ind{b}$ needs not be a $b$-invariant itself. The reason is that $b$-sets are not closed under intersection. Indeed, the conjunction of two formulas with $b$ clauses is not always equivalent to a formula with $b$ clauses, one can only guarantee equivalence to a formula with $2b$ clauses. 

\begin{exa}\label{ex:deadlock-dining-philosopers}
Recall that the deadlocked configurations of the dining philosophers are 
        \[ \mathit{Dead}=\comp{\Sigma^* \fork \, \tilo \, \fork \, \Sigma^*} \cap \comp{\Sigma^* \bork \, \eilo \, \bork  \, \Sigma^*} \cap 
        \comp{\tilo \, \fork \, \Sigma^* \fork } \cap \comp{\eilo \, \bork  \, \Sigma^* \bork }  \ . \] 
We prove $\Ind{1} \cap \mathit{Dead} = \emptyset$, which implies that the dining philosophers are deadlock-free.  Let $C$ be the set of configurations of $((\tilo\mid \eilo)(\fork \mid \bork))^*$ containing no occurrence of $\eilo \fork \eilo$ or $\tilo\, \bork \, \tilo$ as a cyclic word, and without alternation of $\eilo$ and $\tilo$ throughout the entire configuration. In \autoref{ex:reachphil}, we showed that $C$ is the intersection of $1$-invariants, which implies $\Ind{1} \subseteq C$. So it suffices to prove $C \cap \mathit{Dead}=\emptyset$. For this, let $w \in C$. If $|w|\leq 3$ the proof is an easy case distinction, so assume $|w|\geq 4$. 
        We show that $w$ as a cyclic word contains an occurrence of $\fork  \, \tilo \, \fork$ or $\bork \, \eilo \, \bork$, and so it is not a deadlock. 
        If all philosophers are thinking at $w$,  then, since $w$ contains no occurrence of $\tilo \, \bork \, \tilo$, it contains an occurrence of $\fork \, \tilo \, \fork$.

        If some philosopher is eating at $w$, assume for a contradiction that there is a deadlock.
        By this assumption, that philosopher must have a free fork nearby. 
        Without loss of generality, assume that the fork is the next symbol in the configuration.
        Then we are looking at a $\eilo \fork$ fragment. 
        As $\eilo \fork \eilo$ is forbidden,
        the next symbol in the configuration has to be $\tilo$. 
        By assumption of a deadlock, the next symbol is $\bork$.
        We have obtained $\eilo \fork \tilo \bork$ fragment; 
        if we continue, we observe that eating and thinking philosophers 
        always alternate. 
        But this is forbidden by the last pair of invariants;
        thus the initial assumption of deadlock must be false.

Further, for the dining philosophers we have $\reachconf{} = \Ind{3}$. Apart from some corner cases
        (e.g. an unsatisfiable invariant for every odd length),
        the reason is that the $3$-formula 
        \begin{align*}
                & (\vari{\tilo}{i}{\ell}
        \vee\vari{\fork}{i}{\ell}
        \vee\vari{\bork}{i}{\ell}
        \vee \vari{\bork}{i+1}{\ell}
        \vee\vari{\tilo}{i+1}{\ell}
        \vee\vari{\eilo}{i+1}{\ell}
                ) \\ & \qquad \wedge (
        \vari{\bork}{i+1}{\ell} 
        \vee \vari{\tilo}{i+1}{\ell}
        \vee \vari{\eilo}{i+1}{\ell}
        \vee \vari{\tilo}{i+2}{\ell}
        \vee\vari{\fork}{i+2}{\ell} 
        \vee\vari{\bork}{i+2}{\ell} 
                ) \\ & \qquad \wedge (
        \vari{\tilo}{i}{\ell} 
        \vee 
        \vari{\fork}{i}{\ell} 
        \vee 
        \vari{\bork}{i}{\ell} 
        \vee 
        \vari{\fork}{i+1}{\ell} 
        \vee 
        \vari{\tilo}{i+1}{\ell} 
        \vee 
        \vari{\eilo}{i+1}{\ell} 
        \vee 
        \vari{\tilo}{i+2}{\ell}
        \vee 
        \vari{\fork}{i+2}{\ell} 
        \vee 
        \vari{\bork}{i+2}{\ell} 
        )
        \end{align*}
        \noindent  is an inductive $3$-invariant for every $\ell \geq 3$ and every $i$ from $1$ to $\ell$, if we interpret indices cyclically modulo $\ell$.
To verify it, observe that any violation has one of $\eilo \fork *$, $* \fork \eilo$, or $\eilo \bork \eilo$
at the positions $i$ through $i+2$.
For either of the former two situations to arise anew, the middle fork needs to be newly freed — requiring the latter situation as a precondition. Conversely, for the latter situation to arise anew, one of the philosophers need to start eating, picking up a free fork next to the other philosopher, who needs to be already eating.
Thus any violation of the condition requires a violation on the previous step, too, proving inductivity.
        The configurations satisfying this invariant and the inductive $1$-invariants $I_0, \ldots, I_6$ of \autoref{ex:binv} are the reachable configurations 
$\reachconf{} = (\tilo ( \fork \mid \bork\eilo\bork))^* \mid \eilo\bork\tilo ((
  \fork \mid \bork\eilo\bork) \tilo)^*\bork$.

        Note that the $3$-invariant does not rely on the $1$-invariants for its inductiveness;
        we always require the inductive invariants to be inductive on their own independently of each other.
        The invariant is weaker than it could be because we have obtained it from 
        forbidding the patterns $\eilo \fork$, $\fork \eilo$, and $\eilo \bork \eilo$.
        It happens to include some obviously unreachable configurations such as $\fork\tilo\fork\tilo$
        (note that the forks and the philosophers are swapped, the initial state for two philosophers is
        $\tilo \fork \tilo \fork$), but this is not a problem as simpler invariants exclude them.
\end{exa}

\begin{exa}
  \label{ex:ladder}
  We construct an (artificial) family $\{\rts_b \mid b \geq 1 \}$ of RTSs such that $\Ind{b} \supset \reach(\mathcal{R}_b)=\Ind{b+1}$.
  Fix some $b \geq 1$. Let $\rts_b= \{ \set{0,1}, \initial, \transitions\}$ with $\rel{\transitions}$ given by the union of the languages
        \[ \vtuple{0 \\ 0}^{k_1} \vtuple{1 \\ 0}  \vtuple{0\\ 0}^{k_2} \, \left( \vtuple{0 \\ 0} \;\bigg|\;  \vtuple{0 \\ 1} \;\bigg|\; \vtuple{1 \\ 0} \;\bigg|\; \vtuple{1 \\ 1} \right)^{*}\]
  \noindent  for every $k_1, k_2 \in \mathbb{N}$ such that $k_1+k_2=b-1$. 
  Then every transition of the RTS is of the form $u \cdot v \leadsto u' \cdot v'$ where 
  $\size{u} = b = \size{u'}$, the word $u$ contains exactly one $1$,  and the word $u'$ contains only $0$s.
   If we choose $0^{*}$ as the set of initial configurations, then no transition is applicable to any initial configuration, and so $\reachconf{}= 0^{*}$.   
   We show that $\Ind{b} \supset \Ind{b+1}= 0^{*} = \reachconf{}$. 
 \begin{itemize}
 \item $\Ind{b+1}= 0^{*}$. We first claim that, for every
 $i > b$, the formula $\varphi_i= \bigwedge_{j=1}^{b} 0_{j \colon \ell} \wedge 0_{(i-b)\colon \ell}$ is an inductive $b+1$-invariant. 
 It follows immediately from the definitions that $\varphi_i$ is a $b+1$-invariant. To show that $\varphi_i$ is inductive,  observe that every
 configuration $w$ of length $\ell$ satisfying $\varphi_i$ is of the form $w=0^b w'$ for some word $w'$, and 
 so no word $u$ satisfies $w \leadsto u$. This proves the claim. Now, the only configurations satisfying
 $\varphi_i$ for every $i > b$ are those containing only $0$s, and so $\Ind{b+1}= 0^{*}$.
 
\item $0^{b} \; 1 \in \Ind{b}$, and so $\Ind{b} \supset \reachconf{}$.
        Intuitively, we want to show that every $b$-set containing $0^{b+1}$
                 contains \emph{some} word of the form $0^*10^*$,
                 and all such words can reach $0^b1$, 
                 requiring all $b$-invariants to contain $0^b1$.
 It suffices to prove that no  inductive $b$-set $\varphi = \bigwedge_{i=1}^{b} \varphi_{i}$ separates $0^{b+1}$ and $0^{b} \; 1$, i.e., satisfies
  $0^{b+1} \models \varphi$ while $0^{b} \; 1 \not\models \varphi$.
  For this, let us introduce a family of useful words:
  for every $1 \leq i \leq b+1$, let $w_i:= 0^{i-1} \; 1 \; 0^{b - i + 1}$.
  We have $w_i \leadsto 0^{b} \; 1$ and so, since $\varphi$ is inductive and $0^{b} \; 1 \not\models \varphi$, we get $w_i \not\models \varphi$.
  Intuitively, we exploit the fact that $w_{i}$ differs from $0^{b+1}$ in exactly one position.
  This implies that, for every $i$, some clause of $\varphi$, say $\varphi_j$,  satisfies $w_{i} \not\models \varphi_j$ and therefore contains $0_{i:b+1}$ as an atomic proposition.
 So we ``need'' all clauses of the inductive invariant to exclude all $w_{i}$, and have none left to exclude $0^{b} \; 1$.

  More formally, we conduct an induction on $i$ from $1$ to $b$ to prove that there are distinct $j_{1}, \ldots, j_{b}$ such that $\varphi_{j_{i}}$ contains the atomic proposition $0_{i:b+1}$.
  Consider the base case:
  Pick $j_{1}$ such that $w_{1} \not\models \varphi_{j_{1}}$, which exists because $w_{1} \not\models \varphi$.
  $0^{b+1}$ differs from $w_{1}$ only in the first letter, and $\varphi_{j_{1}}$ is a disjunction of atomic propositions which is satisfied by $0^{b+1}$.
  This means that $\varphi_{j_{1}}$ must contain $0_{1:b+1}$ because it is the only atomic proposition that is satisfied by $0^{b+1}$ but not $w_{1}$.
  For the induction step, one observes that $w_{i}$ satisfies the clauses $\varphi_{j_{1}}, \ldots, \varphi_{j_{i-1}}$ by induction hypothesis.
  Therefore, $\varphi$ also contains a new clause $\varphi_{j_{i}}$ such that $w_{i} \not\models \varphi_{j_{i}}$.
  With the same reasoning as for the base case one concludes that $0_{i:b+1}$ is an atomic proposition in $\varphi_{j_{i}}$.
  This concludes the induction and, thus, $0^{b} \; 1 \models \varphi$.

\end{itemize}
\end{exa}

\begin{exa}
\label{ex:ladder2}
As a final example, we construct an RTS such that $\reach$ is a regular language but there is no $b$ such that $\Ind{b}=\reach$. Let $\rts=(\{0,1\}, \initial, \transitions \}$ be the RTS with  $\lang{\initial}= 0^{*}$ and
$$\rel{\transitions}= \vtuple{0 \\ 0}^{*} \; \vtuple{1 \\ 0} \; \vtuple{0 \\ 0}^{*} \; \left( \vtuple{0 \\ 0} \;\bigg|\; \vtuple{0 \\ 1} \;\bigg|\; \vtuple{1 \\ 0} \;\bigg|\; \vtuple{1 \\ 1} \right) \, .$$
        For every length $\ell$, the RTS $\rts$ behaves like the RTS $\rts_{b}$ of \autoref{ex:ladder} for $b = \ell -1$. Therefore, we have $\reach = 0^*$, and so $\reach$ is regular. Further, $\Ind{b} \neq \reach$ for every $b \geq 1$ (as they differ when restricted to the length $b+1$).
\end{exa}

\section{\texorpdfstring{$\Ind{b}$}{Ind\_b} is regular for every \texorpdfstring{$b \geq 1$}{b≥1}}
\label{sec:indinv-is-regular}

\subsection{Encoding \texorpdfstring{$b$}{b}-formulas as \texorpdfstring{$b$}{b}-powerwords.}
\label{sec:indreg:prelim}
We introduce an encoding of $b$-formulas. We start with some examples. Assume
$\mathcal{R}$ is an RTS with $\Sigma=\{a,b,c\}$. We consider formulas over $\mathit{AP}_3$, i.e., over the atomic propositions $\{ \vari{a}{1}{3},  \vari{a}{2}{3},  \vari{a}{3}{3}, \vari{b}{1}{3}, \vari{b}{2}{3}, \vari{b}{3}{3}, \vari{c}{1}{3}, \vari{c}{2}{3}, \vari{c}{3}{3} \}$. 

We encode the $1$-formula $(\vari{a}{1}{3} \vee \vari{a}{2}{3})$ as the word $\{a\} \, \{a\} \, \emptyset$ of length three over the alphabet $2^\Sigma$. Intuitively, $\{a\} \, \{a\} \, \emptyset$ stands for the words of length $3$ that have an $a$ in their first or second position. Similarly, we encode $(\vari{a}{1}{3} \vee \vari{b}{1}{3} \vee \vari{b}{3}{3})$ as $\{a,b\} \, \emptyset \, \{b\}$. Intuitively, $\{a,b\} \, \emptyset \, \{b\}$ stands for the set of words of length $3$ that have $a$ or $b$ as the first letter, or $b$ as the third letter. Since $2^\Sigma$ is the powerset of $\Sigma$, we call words over $2^\Sigma$ \emph{powerwords}.

Consider now the $2$-formula $(\vari{a}{1}{3} \vee \vari{b}{1}{3} \vee \vari{a}{2}{3}) \wedge (\vari{b}{1}{3} \vee \vari{b}{3}{3} \vee \vari{c}{3}{3})$. 
We put the encodings of its clauses ``on top of each other''. Since the encodings of $(\vari{a}{1}{3} \vee \vari{b}{1}{3} \vee \vari{a}{2}{3})$ and  $(\vari{b}{1}{3} \vee \vari{b}{3}{3} \vee \vari{c}{3}{3})$ are $\{a,b\} \,\{a\} \, \emptyset$ and $\{b\}\, \emptyset \, \{b, c\}$,  respectively, we encode the formula as the word
\[ \vtuple{ \{a,b\} \\[0.1cm] \{b\} } \vtuple{ \{a\} \\[0.1cm] \emptyset }
\vtuple{ \emptyset \\[0.1cm] \{b,c\} } \]

\noindent of length three over the alphabet $2^\Sigma \times 2^\Sigma = (2^\Sigma)^2$. We call such a word a \emph{$2$-powerword}. 
Similarly, we encode a $b$-formula over $\mathit{AP}_3$ as a \emph{$b$-powerword} of length three over the alphabet $(2^\Sigma)^b$. In the following, we overload $\varphi$  to denote both a formula and its encoding as a $b$-powerword, and, for example, write

\begin{equation*}
  \varphi = \vtuple{ X_{11}\\[0.1cm] \cdots \\[0.1cm] X_{b1} } \cdots
  \vtuple{ X_{1\ell}\\[0.1cm] \cdots \\[0.1cm]  X_{b\ell} }  \quad
  \mbox{ instead of  } \quad \varphi = \bigwedge_{i=1}^b \bigvee_{j=1}^\ell \bigvee_{a \in X_{ij}} \vari{a}{i}{\ell}
\end{equation*}
where $X_{ij} \subseteq \Sigma$. Intuitively, each row $X_{i1} \cdots
X_{i\ell}$ encodes one clause of $\varphi$. We also write $\varphi = \varphi[1]
\cdots \varphi[\ell]$ where $\varphi[i] \in (2^\Sigma)^b$ denotes the $i$-th
letter of the $b$-powerword encoding $\varphi$. 

Now we show a simple lemma, which, however, provides the key to our results. For every fixed $b$, the satisfaction relation $w \models \varphi$ between configurations 
and $b$-formulas is regular:

\begin{lem}
\label{lem:satisfaction}
Let $\Sigma$ be an alphabet, $b \geq 1$, and $\Gamma=(2^\Sigma)^b$. There exists a deterministic transducer $\interpretation_b$ over the alphabet $\Sigma \times \Gamma$ with $2^b$ states 
such that $L(\interpretation) = \{ \tuple{w, \varphi} \in (\Sigma \times \Gamma)^* \mid w \models \varphi\}$.
\end{lem}
\begin{proof}
A configuration $w = w[1] w[2] \cdots w[\ell]$ satisfies a b-formula 
\[ \varphi = \vtuple{ X_{11} \\ \ldots \\ X_{b1} } \vtuple{ X_{12} \\ \ldots \\ X_{b2}} \cdots  
\vtuple{ X_{1\ell}  \\ \ldots \\ X_{b\ell}} \]
\noindent where $X_{i,j} \subseteq \Gamma$ if{}f for every index $1 \leq i \leq b$ there exists 
a position $1 \leq j\leq \ell$ in the word $w$
such that $w[j] \in X_{i j}$.  We define a transducer $\interpretation_b$ over $\Sigma \times \Gamma$ that accepts $\tuple{w, \varphi} $ if{}f this condition holds.
The transducer reads the word
\[ \vtuple{ w[1] \\[0.1cm] X_{11} \\ \ldots \\ X_{b1} } \vtuple{ w[2] \\[0.1cm] X_{12}  \\ \ldots \\ X_{b2}} \cdots  
\vtuple{ w[\ell] \\[0.1cm] X_{1\ell}  \\ \ldots \\ X_{b\ell}} \]
\noindent storing in its state the set of indices $i \in \{1, \ldots, b\}$ for which the position $j$ has already been found. 
So the states of the transducer are the subsets of $\{1, \ldots, b\}$, and, given two states $S, S'$, there is a transition from $S$
to $S'$ labeled by the tuple $\tuple{ a , X_{1}, \ldots , X_{b} }$ if{}f $S \subseteq S'$ and $a \in X_{i}$ holds for every $i \in S' \setminus S$.
The initial state is the empty set, indicating that no position has been found yet, and the unique final state is the set $\{1, \ldots, b\}$, indicating that all positions have been found.
\end{proof}

\subsection{Representations and interpreters}   
\label{subsec:representations-interpreters}
The transducer $\interpretation_b$ of \autoref{lem:satisfaction} recognizes the satisfaction relation
$w \models \varphi$ between configurations and $b$-formulas.  We generalize this idea. Consider an \emph{arbitrary} transducer $\interpretation$ over the alphabet $\Sigma \times \Gamma$,
where $\Gamma$ is some \emph{arbitrary} alphabet. Now, look at a word $W \in \Gamma^*$ (we use capital letters $W,V,U, \ldots$ to denote such words) as a \emph{representation}
of the set of configurations $\{ w \in \Sigma^* \mid \tuple{w, W} \in \lang{\interpretation} \}$. Intuitively, $W$ is a \emph{name} standing for this set of configurations, and the transducer $\interpretation$ \emph{interprets} the meaning of $W $. 
For this reason,  we call $\interpretation$ an \emph{interpreter}. 

\begin{defi}[Interpretation]
Let $\rts = \tuple{\Sigma, \initial, \transitions}$ be a regular transition system. An \emph{interpretation} is a pair $\tuple{\Gamma, \interpretation}$ where $\Gamma$ is an alphabet
and $\interpretation$ is a deterministic transducer over $\Sigma \times \Gamma$, called the \emph{interpreter}\footnote{Observe that we require the transducer to be deterministic (as is the case in \autoref{lem:satisfaction}).}. We call words over $\Gamma$ \emph{representations}, and use capital letters $W, V, U, \ldots$ to denote them. $W \in \Gamma^*$ \emph{represents} or \emph{stands for} the set of all configurations $w \in \Sigma^*$ such that $\tuple{w, W}$ is accepted by $\interpretation$. We write $w \models_{\interpretation} W$ to denote that $w$ is one the words represented by $W$.
\end{defi}

\begin{exa}
The transducer of \autoref{lem:satisfaction} is a particular interpreter with alphabet $\Gamma =(2^\Sigma)^b$.  For example, for $b=2$ the transducer
interprets the word \[ \vtuple{ \{a,b\} \\[0.1cm] \{b\} } \vtuple{ \{a\}  \\[0.1cm] \emptyset }
\vtuple{ \emptyset  \\[0.1cm] \{b,c\} } \]
as the set of configurations that satisfy the formula $(\vari{a}{1}{3} \vee \vari{b}{1}{3} \vee \vari{a}{2}{3}) \wedge (\vari{b}{1}{3} \vee \vari{b}{3}{3} \vee \vari{c}{3}{3})$.
\end{exa}
\noindent 
In the rest of the section, we define the set $\Ind{\interpretation}$ for an arbitrary interpreter $\interpretation$, and prove that it is regular. Instantiating $\interpretation$ as the transducer $\interpretation_b$ of \autoref{lem:satisfaction}, we obtain 
as a corollary that $\Ind{b}$ is regular for every $b \geq 1$. In order to define $\Ind{\interpretation}$, observe that an interpreter $\interpretation$ may interpret some representations as inductive sets of configurations, and others as non-inductive sets.  
We define the set of all representations that are inductive.

\begin{defi}
Let $\rts = \tuple{\Sigma, \initial, \transitions}$ be a regular transition system and let $\interpretation$
be an interpreter. A representation $W \in \Gamma^*$ is \emph{inductive} if $u \leadsto w$ and $u \models_{\interpretation} W$
implies $w \models_{\interpretation} W$. We define the set 
\begin{equation*}
  \Inductive = \set{W \in \Gamma^{*} \mid \text{ for every } u, w \in \Sigma^*, \text{ if } u \leadsto w \text{ and } u \models_{\interpretation} W \text{, then } w \models_{\interpretation} W}.
\end{equation*}
We write $u \reaches_{\interpretation} w$, and say that $w$ is \emph{potentially reachable} from $u$ with respect to $\interpretation$ if for every inductive representation $W \in \Inductive$, if $u \models_{\interpretation} W$, then $w \models_{\interpretation}{W}$ too. Further, we define 
\begin{equation*}
\Ind{\interpretation} := \{ w \in \Sigma^* \mid u \prelint w \text{ for some initial configuration } u \}  .
\end{equation*}
\end{defi}

We have the following fact:

\begin{fact}
$\Ind{\interpretation} \supseteq \reach$ for every interpreter $\interpretation$, i.e., $\Ind{\interpretation}$ is an overapproximation of the set of reachable configurations.
\end{fact}
\begin{proof}
If $w \in \reach$, then $v  \reaches w$ holds for some initial configuration $v$. We prove $v \prelint w$, which implies $w \in \Ind{\interpretation}$. Since $v  \reaches w$, every inductive set 
containing $v$ contains $w$ as well. So, in particular, for every $W \in \Inductive$, if $v \models_\interpretation W$, then $w \models_\interpretation W$.
\end{proof}
\noindent 
To illustrate the virtue of this generalization, consider another interpretation for the RTS from \autoref{ex:ladder2}. Abstractly speaking, we consider a set of atomic propositions and enforce that none of these propositions is true. The interpreter $\interpretation$ for this is

\begin{center}
\begin{tikzpicture}
  \node[initial, initial text=, initial left, state, accepting] (0) {\phantom{x}};
  \node[state, right=of 0] (1) {\phantom{x}};

  \draw[->, loop above] (0) to node {$M$} (0);
  \draw[->, above] (0) to node {$H$} (1);
  \draw[->, loop above] (1) to node {$M, H$} (2);
\end{tikzpicture}
\end{center}
where $M$ are all pairs $\tuple{\sigma, U}$ with $\sigma \notin U$ while $H$ are all pairs with $\sigma \in U$. One can quickly verify that here $\set{1}^{*}\subseteq \Inductive$ and, thus, $\mathit{Reach} = 0^{*} = \Ind{\interpretation}$ as there is no $1$ in any configuration -- a fact that is true initially, and, since then there is no transition applicable, throughout every step.

In the rest of the section, we prove that $\Ind{\interpretation}$ is an effectively regular set of configurations for any interpreter $\interpretation$.  As a first step, 
we show that $\Inductive$ is a regular set of representations. This follows immediately from the following proposition, proving that the
complement of  $\Inductive$ is regular.

\begin{lem}
 \label{lem:inv-inductive-inv}
 Let $\rts = \tuple{\Sigma, \initial, \transitions}$ be an RTS where $\transitions$ has $n_\transitions$ states, and let $\interpretation$
 be an interpreter with $n_\interpretation$ states.  One can effectively compute an NFA with at most $n_{\mathcal{T}} \cdot n_{\interpretation}^2$ states recognizing the set of representations $\overline{\Inductive}$.
\end{lem}
\begin{proof}
  Let $\mathit{Id}_\Gamma = \{ \tuple{W, W} \mid W \in \Gamma^*\}$ be the identity relation on $\Gamma$. By definition of $\Inductive$, we have
  \begin{align*}
  & \overline{\Inductive}  \\
   = &  \set{W\in\Gamma^{*}\mid \exists u , w \in \Sigma^* \text{ s.t. } u \leadsto w, u \models_{\interpretation} W \text{ and } w\not\models_{\interpretation} W } \\
   = & \set{W\in\Gamma^{*}\mid \exists u , w \in \Sigma^* \text{ s.t. } \tuple{u, w} \in \varrel{\transitions}, \tuple{u, W} \in \varrel{\interpretation}, \tuple{w, W} \in \comp{\varrel{\interpretation}} }   \\
   = & \set{W\in\Gamma^{*}\mid \exists u , w \in \Sigma^* \text{ s.t. } \tuple{W,u} \in {\varrel{\interpretation}}^{-1}, \tuple{u, w} \in \rel{T}, \tuple{w, W} \in \comp{\varrel{\interpretation}} }  \\
   = & \set{   W\in\Gamma^{*}  \mid   \tuple{W, W} \in \varrel{\interpretation}^{-1} \circ \rel{T} \circ \comp{\varrel{\interpretation}}   }  \\
   = & \projg{ \left( \big({\varrel{\interpretation}}^{-1} \circ \rel{T} \circ \comp{\varrel{\interpretation}} \big) \cap \mathit{Id}_\Gamma \right) }{1} .
  \end{align*}
  Since $\interpretation$ is deterministic, there are transducers for $\rel{\interpretation}$ and $\compl{\rel{\interpretation}}$ with $\O(n_{\interpretation})$ states. Applying 
  \autoref{prop:transducerops}(1) twice and \autoref{prop:transducerops}(3), we obtain an NFA for $\comp{\Inductive} $ with at most $n_{\mathcal{T}} \cdot n_{\interpretation}^2$ states.
\end{proof}
\noindent 
Using standard automata constructions, we get immediately:
\begin{lem}
 \label{lem:inv-inductive}
 Let $\rts = \tuple{\Sigma, \initial, \transitions}$ be an RTS where $\transitions$ has $n_\transitions$ states, and let $\interpretation$
  be an interpreter with $n_\interpretation$ states.  One can effectively compute an DFA with at most $2^{n_{\mathcal{T}} \cdot n_{\interpretation}^2}$ states recognizing the set of representations $\Inductive$.
\end{lem}

\noindent 
We prove that the potential reachability relation is effectively regular, that is, the relation is recognized by a transducer that can be effectively constructed from 
$\mathcal{T}$ and $\interpretation$. Again, we show that its complement is regular.

\begin{lem}
 \label{thm:regular-abstraction-inv}
 Let $\rts = \tuple{\Sigma, \mathcal{I}, \mathcal{T}}$ be a regular transition system with $n_{\mathcal{T}}$ states, and let $\interpretation$ be an interpreter  with  $n_{\interpretation}$ states. One can effectively compute a nondeterministic transducer  with at most $n_{\interpretation}^2 \cdot 2^{n_{\mathcal{T}} \cdot n_{\interpretation}^2}$ states recognizing $\compl{\prelint}$.
\end{lem}
\begin{proof}
By definition, we have
  \begin{align*}
    & \compl{\prelint}  \\
    = &  \set{\tuple{u, w} \in (\Sigma \times \Sigma)^*  \mid \exists W \in \Inductive \text{ s.t. } u \models_\interpretation W  \text{ and }  w \not\models_\interpretation W} \\
    = &  \set{\tuple{u, w} \in (\Sigma \times \Sigma)^* \mid \exists W \in \Inductive \text{ s.t. } \tuple{u, W} \in \varrel{\interpretation} \text{ and }  \tuple{W, w} \in \comp{\varrel{\interpretation}^{-1}} }   .
\end{align*}
Let $S = \{ \tuple{W,W} \mid W \in \Inductive \}$. We then have $\compl{\prelint} = \big({\varrel{\interpretation}} \circ S \circ \compl{\varrel{\interpretation}^{-1}} \big)$.
        Since $\interpretation$ is deterministic, there is a transducer for $\compl{\varrel{\interpretation}^{-1}}$ with $\O(n_{\interpretation})$ states. 
  Further, by \autoref{lem:inv-inductive-inv} there is a DFA for $\Inductive$ with at most $2^{n_{\mathcal{T}} \cdot  n_{\interpretation}^2}$ states, and so a transducer for
  $S$ with the same number of states. Applying \autoref{prop:transducerops}, we can construct a nondeterministic transducer for $\compl{\prelint}$ with at most
  $n_{\interpretation}^2 \cdot  2^{n_{\mathcal{T}} \cdot  n_{\interpretation}^2}$  states.
\end{proof}
\noindent 
Again, we get:
\begin{lem}
 \label{thm:regular-abstraction}
  Let $\rts = \tuple{\Sigma, \mathcal{I}, \mathcal{T}}$ be a regular transition system with $n_{\mathcal{T}}$ states, and let $\interpretation$ be an interpreter  with  $n_{\interpretation}$ states. One can effectively compute a deterministic transducer  with at most $2^{2^{\log (n_{\interpretation}^2) \cdot n_{\mathcal{T}} \cdot n_{\interpretation}^2}}$ states recognizing $\prelint$.
\end{lem}

\noindent 
We combine the previous results to show that, given an RTS $\mathcal{R}$ and an interpretation $\interpretation$, the
set $\Ind{\interpretation}$ of potentially reachable configurations is recognized by an NFA with double exponentially many states in $\mathcal{T}$ and $\interpretation$.
\begin{thm}
\label{thm:doubexp}
Let $\mathcal{R} = \tuple{\Sigma, \initial, \transitions}$ be an RTS, and let $\interpretation$ be an interpretation. Let
$n_I$, $n_T$, and $n_\interpretation$ be the number of states of $\initial$, $\transitions$, and $\interpretation$, respectively. 
Then $\Ind{\interpretation}$ is recognized by an NFA with at most $n_\mathcal{I} \cdot 2^{2^{\log (n_{\interpretation}^2) \cdot n_{\mathcal{T}} \cdot n_{\interpretation}^2}}$ states.
\end{thm}
\begin{proof}
We have $\Ind{\interpretation}   =   \varlang{\initial} \circ (\reaches_\interpretation)$ by the definition of $\Ind{\interpretation}$.
By \autoref{thm:regular-abstraction}, $\reachesb{\interpretation}$ is recognized by a deterministic transducer  with $2^{2^{\log (n_{\interpretation}^2) \cdot n_{\mathcal{T}} \cdot n_{\interpretation}^2}}$ states.
Apply \autoref{prop:transducerops}.
\end{proof}
\noindent 
We apply \autoref{thm:doubexp} to the interpreter of $b$-formulas given in \autoref{lem:inv-inductive}, and obtain:

\begin{cor}
\label{cor:doubexp}
Let $\mathcal{R} = \tuple{\Sigma, \initial, \transitions}$ be an RTS. Let $b \geq 1$, and let $f(n_{\mathcal{T}}, b) := 2\cdot n_{\mathcal{T}}  \log n_{\mathcal{T}} \cdot 2^{b+1}$.   Then $\Ind{b}$ (\autoref{def:indb}) is recognized by a DFA with at most $n_\mathcal{I} \cdot 2^{2^{f(n_{\mathcal{T}}, b)}}$ states.
\end{cor}
\begin{proof}
By \autoref{lem:inv-inductive}, there is a deterministic transducer  with at most $2^{2^{(2 n_{\mathcal{T}} \log n_{\mathcal{T}} \cdot n_{\interpretation})}}$ states recognizing $\reachesb{\interpretation}$. So, there is a deterministic transducer with at most $2^{2^{f(n_{\mathcal{T}}, b)}}$ states for $\reachesb{b}$.
Apply \autoref{prop:transducerops}(2).
\end{proof}

\noindent 
Given an instance $\rts$, $\mathcal{U}$ of the safety verification problem and a fixed $b \geq 0$, if the set $\Ind{b}$ satisfies $\Ind{b} \cap
\lang{\mathcal{U}} = \emptyset$, then $\mathcal{R}$ is safe. By \autoref{cor:doubexp}, deciding whether $\Ind{b} \cap \lang{\mathcal{U}} =
\emptyset$ is in \textsc{EXSPACE} for every fixed $b$. Indeed, the theorem and
its proof show that there is a DFA recognizing $\Ind{b} \cap \lang{\mathcal{U}} $ such that one can guess an accepting path of it,
state by state, using exponential space. Indeed, for fixed $b$, the interpreter of \autoref{lem:satisfaction} has a constant number of states $n_\interpretation$.
Therefore, storing one state of the transducer of \autoref{thm:regular-abstraction} takes polynomial space, and storing one state of the transducer of 
\autoref{thm:doubexp} takes exponential space.  Currently, we do not know if there is a $b$ such
that the problem is \textsc{EXSPACE}-complete for every $b' \geq b$. In the next 
two sections, we show that for $b=1$ the problem is actually \textsc{PSPACE}-complete.

\section{Deciding  \texorpdfstring{$\Ind{1} \cap \lang{\mathcal{U}} = \emptyset$}{Ind\_1 ∩ U = ∅} is in \textsc{PSPACE}}
\label{sec:pspacecomp}

We give a non-deterministic polynomial-space algorithm that decides $\Ind{1} \cap \lang{\mathcal{U}} = \emptyset$.
As a byproduct, we show that $\Ind{1}$ is recognized by an NFA with a single exponential number of states.

\smallskip
As a running example for the following construction we are introducing a basic token passing protocol:
There is a line of agents.
Any single agent either holds a token ($t$) or not ($n$).
The initial language is $t \; n^*$; that is, initially there is exactly one token.
The transitions of the system allow the token to be passed down the line or, if the token already is at the last position, to be passed back to the front.
One can understand this protocol as a token passing algorithm in a circle of agents.
The transducer for the transitions looks as follows:

  \begin{center}
    \begin{tikzpicture}
      \node[state, initial, initial text=, initial left] (q0) {$q_{0}$};
      \node[state, above right=of q0] (q1) {$q_{1}$};
      \node[state, right=of q1] (q2) {$q_{2}$};
      \node[state, right=of q2, accepting] (q3) {$q_{3}$};
      \node[state, below right=of q0] (q4) {$q_{4}$};
      \node[state, right=of q4, accepting] (q5) {$q_{5}$};

      \draw[->] (q0) to node[above] {$\vtuple{n \\ n}$} (q1);
      \draw[->] (q1) to node[above] {$\vtuple{t \\ n}$} (q2);
      \draw[->, near end] (q0) to node[below] {$\vtuple{t \\ n}$} (q2);
      \draw[->] (q2) to node[above] {$\vtuple{n \\ t}$} (q3);
      \draw[->, loop above] (q1) to node[above] {$\vtuple{n \\ n}$} (q1);
      \draw[->, loop above] (q3) to node[above] {$\vtuple{n \\ n}$} (q3);
      \draw[->, loop above] (q4) to node[above] {$\vtuple{n \\ n}$} (q4);
      \draw[->] (q0) to node[above] {$\vtuple{n \\ t}$} (q4);
      \draw[->] (q4) to node[above] {$\vtuple{t \\ n}$} (q5);
    \end{tikzpicture}
  \end{center}

\smallskip
We fix an RTS $\rts = ( \Sigma, \initial, \transitions)$ for the rest of the section.
$1$-formulas have a special property: since the disjunction of two clauses is again a clause, the disjunction of two $1$-formulas is again a $1$-formula.  This allows us to define the \emph{separator} of a configuration $w$.

\begin{defi}
The \emph{separator} of a configuration $w$, denoted $\Sep{w}$, is the union of
  all inductive $1$-sets \emph{not} containing $w$.
\end{defi}

We characterize membership of $w$ in $\Ind{1}$ in terms of its separator:

\begin{lem}
\label{lem:indchar}
For every configuration $w$, its separator $\Sep{w}$ is an inductive $1$-set. Further $w \in \Ind{1}$ if{}f 
$\Sep{w}$ is not an invariant.
\end{lem}
\begin{proof}
Since inductive sets are closed under union, $\Sep{w}$ is inductive. Since the disjunction of two clauses is again a clause, the union of two $1$-sets of configurations is also a $1$-set, and so $\Sep{w}$ is an inductive $1$-set. For the last part, we prove that 
$w \notin \Ind{1}$ if{}f $\Sep{w}$ is an invariant. Assume first $w \notin \Ind{1}$. Then some inductive $1$-invariant does not contain $w$.
        Since, by definition, $\Sep{w}$ contains this invariant, $\Sep{w}$ is also an invariant. Assume now that $\Sep{w}$ is an invariant. Then $\Sep{w}$ is an inductive $1$-invariant, and so $\Sep{w} \supseteq \Ind{1}$. Since
$w \notin \Sep{w}$, we get $w \notin \Ind{1}$.
\end{proof}
\noindent 
Our plan for the rest of the section is as follows:
\begin{itemize}
\item We introduce the notion of a separation table for a configuration. (\autoref{def:separationtable})
\item We show that, given a configuration $w$ and a separation table $\tau$ for $w$, we can construct a $1$-formula 
$\Sepf{w}$ such that $\lang{\Sepf{w}} = \Sep{w}$. (\autoref{lem:sepchar})
\item We use this result to define a transducer $T_{sep}$ over $\Sigma \times 2^\Sigma$ 
that accepts a word $\tuple{w, \varphi}$ if{}f $\varphi = \Sepf{w}$. (\autoref{prop:transducer})
\item We use $T_{sep}$ and \autoref{prop:transducerops} to define an NFA over $\Sigma$ that accepts a configuration $w$ if{}f  $\Sep{w}$ is not an invariant, and so, by \autoref{lem:indchar}, if{}f $w \in \Ind{1}$. (\autoref{thm:autforInd1})
\end{itemize}

\noindent 
We present a characterization of $\Sep{w}$ in terms of \emph{tables}. Given a
transition $s \leadsto t$, we call $s$ and $t$ the \emph{source} and
\emph{target} of the transition, respectively.  A \emph{table} of length $\ell$
is a sequence $\tau = s_1 \leadsto t_1, \ldots, s_n \leadsto t_n$ of
transitions of $\mathcal{R}$ (not necessarily distinct), all of length $\ell$.\footnote{
We call it a table because we visualize $s_1, t_1, \ldots, s_n,t_n$ as a matrix 
with  $2n$ rows and $\ell$ columns.}
We define the \emph{separation tables} of a
configuration $w$.

\begin{defi}
\label{def:separationtable}
Let $w$ be a configuration  and let $\tau = s_1 \leadsto t_1, \ldots, s_n \leadsto t_n$ be a table, both of length $\ell$.
For every $j \in [1, \ell]$, let $\Incl{w,\tau}[j] = \{w[j], s_1[j], \ldots, s_n[j]\}$ be the set of letters at position $j$ of $w$ and of the source configurations $s_1, \ldots, s_n$ of the table.
\begin{itemize}
\item $\tau$ is \emph{consistent with $w$} if for every $i \in [1, n], j \in [1, \ell]$, either $t_i[j]=w[j]$ or $t_i[j] = s_{i'}[j]$ for some $i'<i$.\\  
(Intuitively: $\tau$ is consistent with $w$ if for every position and every target configuration of the table, the letter of the target at that position is either the letter of $w$, or the letter of some earlier source configuration, with the choice for different positions made independently.)
\item $\tau$ is \emph{complete for $w$} if every table $\tau, \; s \leadsto t$ consistent with $w$ satisfies $s[j] \in \Incl{w,\tau}[j]$ for every $j \in [1, \ell]$. \\
(Intuitively: $\tau$ is complete for $w$ if it cannot be extended by a transition that maintains consistency and introduces a new letter.)
\end{itemize}
A table is a \emph{separation table of $w$} if  it is consistent with and complete for $w$.
\end{defi}

Overall, a separation table represents the following logic.
We want an inductive 1-invariant that does not contain $w$.
We think of the words that must be outside the invariant.
As a 1-set is defined via a disjunction of atomic propositions,
effectively the question is which letters are excluded at each position.
Naturally, all the letters of $w$ are excluded for their corresponding positions.
Moreover, any transition preimage of a word constructed of excluded letters
also has to be an excluded word due to the inductive property.
Letters of such a preimage are also excluded at their corresponding positions.
Note that this means that we use the transitions backwards: we start with $w$, 
then try to add something that can reach $w$ in one step, etc.
For instance in the running example this manifests as follows:
  Consider the (reachable) configuration $n \; t \; n \; n \; n \; n$ for our running example.
  The largest statement that is not satisfied by this configuration is $\set{t} \; \set{n} \; \set{t} \; \set{t} \; \set{t} \; \set{t}$ since it contains at every position all letters but the one that is at the same position in the original configuration.
  We now demonstrate how this statement is refined to become inductive.
  For this, we show a series of statements and transitions below such that the transition refines the previous statement to the next.
  In the following table we mark statements with $\bullet$ and the refining transitions with $\triangleright$.
  Moreover, we mark in red the atoms of the statements that are removed in each step and in the transitions the reason why they are removed.

  \begin{center}
  \begin{tabular}{ccccccc}
    $\bullet$ &
    $\set{\textcolor{red}{t}}$ &
    $\set{\textcolor{red}{n}}$ &
    $\set{t}$ &
    $\set{t}$ &
    $\set{t}$ &
    $\set{t}$ \\
    $\triangleright$ &
    $\vtuple{\textcolor{red}{t} \\ n}$ &
    $\vtuple{\textcolor{red}{n} \\ t}$ &
    $\vtuple{n \\ n}$ &
    $\vtuple{n \\ n}$ &
    $\vtuple{n \\ n}$ &
    $\vtuple{n \\ n}$ \\
    $\bullet$ &
    $\emptyset$ &
    $\emptyset$ &
    $\set{t}$ &
    $\set{t}$ &
    $\set{t}$ &
    $\set{\textcolor{red}{t}}$ \\
    $\triangleright$ &
    $\vtuple{n \\ t}$ &
    $\vtuple{n \\ n}$ &
    $\vtuple{n \\ n}$ &
    $\vtuple{n \\ n}$ &
    $\vtuple{n \\ n}$ &
    $\vtuple{\textcolor{red}{t} \\ n}$ \\
    $\bullet$ &
    $\emptyset$ &
    $\emptyset$ &
    $\set{t}$ &
    $\set{t}$ &
    $\set{\textcolor{red}{t}}$ &
    $\emptyset$ \\
    $\triangleright$ &
    $\vtuple{n \\ n}$ &
    $\vtuple{n \\ n}$ &
    $\vtuple{n \\ n}$ &
    $\vtuple{n \\ n}$ &
    $\vtuple{\textcolor{red}{t} \\ n}$ &
    $\vtuple{n \\ t}$ \\
    $\bullet$ &
    $\emptyset$ &
    $\emptyset$ &
    $\set{t}$ &
    $\set{\textcolor{red}{t}}$ &
    $\emptyset$ &
    $\emptyset$ \\
    $\triangleright$ &
    $\vtuple{n \\ n}$ &
    $\vtuple{n \\ n}$ &
    $\vtuple{n \\ n}$ &
    $\vtuple{\textcolor{red}{t} \\ n}$ &
    $\vtuple{n \\ t}$ &
    $\vtuple{n \\ n}$ \\
    $\bullet$ &
    $\emptyset$ &
    $\emptyset$ &
    $\set{\textcolor{red}{t}}$ &
    $\emptyset$ &
    $\emptyset$ &
    $\emptyset$ \\
    $\triangleright$ &
    $\vtuple{n \\ n}$ &
    $\vtuple{n \\ n}$ &
    $\vtuple{\textcolor{red}{t} \\ n}$ &
    $\vtuple{n \\ t}$ &
    $\vtuple{n \\ n}$ &
    $\vtuple{n \\ n}$ \\
    $\bullet$ &
    $\emptyset$ &
    $\emptyset$ &
    $\emptyset$ &
    $\emptyset$ &
    $\emptyset$ &
    $\emptyset$ \\
  \end{tabular}
  \end{center}

Consistency with $w$ means that the table contains the transitions
implementing this approach.
Completeness means that there are no transitions left to add 
to exclude more letters.
Strictly speaking, we do not require each transition in the table
to exclude a new letter, but adding transitions without excluding 
any new letters is useless.

For our running example a separation table for the configuration $n \; t \; n \; n \; n \; n$ is
\begin{center}
  \begin{tabular}{ccccccc}
    $\triangleright$ &
    $\vtuple{t \\ n}$ &
    $\vtuple{n \\ t}$ &
    $\vtuple{n \\ n}$ &
    $\vtuple{n \\ n}$ &
    $\vtuple{n \\ n}$ &
    $\vtuple{n \\ n}$ \\
    $\triangleright$ &
    $\vtuple{n \\ t}$ &
    $\vtuple{n \\ n}$ &
    $\vtuple{n \\ n}$ &
    $\vtuple{n \\ n}$ &
    $\vtuple{n \\ n}$ &
    $\vtuple{t \\ n}$ \\
    $\triangleright$ &
    $\vtuple{n \\ n}$ &
    $\vtuple{n \\ n}$ &
    $\vtuple{n \\ n}$ &
    $\vtuple{n \\ n}$ &
    $\vtuple{t \\ n}$ &
    $\vtuple{n \\ t}$ \\
    $\triangleright$ &
    $\vtuple{n \\ n}$ &
    $\vtuple{n \\ n}$ &
    $\vtuple{n \\ n}$ &
    $\vtuple{t \\ n}$ &
    $\vtuple{n \\ t}$ &
    $\vtuple{n \\ n}$ \\
    $\triangleright$ &
    $\vtuple{n \\ n}$ &
    $\vtuple{n \\ n}$ &
    $\vtuple{t \\ n}$ &
    $\vtuple{n \\ t}$ &
    $\vtuple{n \\ n}$ &
    $\vtuple{n \\ n}$ \\
  \end{tabular}
\end{center}

Observe that every configuration $w$ has at least one separation table. If there are no transitions with target $w$, then the empty table with no transitions is a separation table. Otherwise, starting with any transition $s \leadsto w$, we repeatedly add a  transition, maintaining consistency and introducing at least one new letter until no such transition exists. Lack of such transitions implies completeness of the table. This procedure terminates---there are only finitely many transitions between configurations of a fixed length---and yields a separation table.

The next lemma shows how to compute a $1$-formula $\Sepf{w}$ such that $\lang{\Sepf{w}} = \Sep{w}$ from \emph{any} separation table $\tau$ of $w$.

\begin{lem}
\label{lem:sepchar}
Let $\tau$ be any separation table for a configuration $w$ of length $\ell$. 
Then $\Sep{w}$ is the set of all configurations $z \in \Sigma^\ell$ such that  $z[j] \notin \Incl{w,\tau}[j]$ for some $j \in [1, \ell]$. In particular, $\Sep{w}$ is the language of the $1$-formula 
        \[ \Sepf{w}:= \bigvee_{j=1}^\ell \left( \bigvee_{a \notin \Incl{w,\tau}[j]} \vari{a}{j}{\ell} \right)\]
\noindent or, in the powerword encoding, of the formula 
        \[ \Sepf{w}= \comp{\Incl{w,\tau}[1]} \, \cdots \, \comp{\Incl{w,\tau}[\ell]} \ . \]
\end{lem}
\noindent 
Before we prove this lemma, observe that the separation table for our running example from before would give $\Sepf{n \; t \; n \; n \; n \; n}$ as the powerword $\emptyset \; \emptyset \; \emptyset \; \emptyset \; \emptyset \; \emptyset$:
In particular, the first letter is $\emptyset$ because it does not contain $n$ as the configuration starts with the letter $n$ but it also does not contain $t$ as the first transition ``removes'' it.

\begin{proof}
        We claim that $\Sepf{w}$ denotes an inductive $1$-set not containing $w$. That $\Sepf{w}$ denotes a $1$-set
        not containing $w$ follows immediately from the definition. To see that $\Sepf{w}$ denotes an inductive set, assume there is a transition $s \leadsto t$ such that $s \models \Sepf{w}$ and $t \not\models \Sepf{w}$.
        Since $t \not\models \Sepf{w}$, 
        all the letters of $t$ are excluded at their corresponding positions, 
        thus the transition $s\leadsto t$ can be added to the table $\tau$ yielding
        the table $\tau, s\leadsto t$ consistent with $w$. Since $s \models\Sepf{w}$, $\tau$ is not complete, contradicting the assumption. 

        Let us now prove that $\Sepf{w}$ denotes the \emph{largest} inductive $1$-set
not containing $w$.
Using the powerword encoding, it is enough to prove that for every position $j$
        and every letter  $x\in\Incl{w,\tau}[j]$,
each inductive 1-set specified by a powerword containing $x$ at position $j$
contains the word $w$.

Consider the prefixes of $\tau$: the empty sequence $\tau_0$,
        and the sequences $\tau_1 = s_1\leadsto{}t_1$,
        $\tau_2=s_1\leadsto{}t_1, s_2\leadsto{}t_2$, up to 
        $\tau_n=\tau=s_1 \leadsto t_1, \ldots, s_n \leadsto t_n$. All of them are tables consistent with $w$.
        By construction, for every $j$, we have 
        $\Incl{w,\tau_0}[j]\subseteq\Incl{w,\tau_1}[j]\subseteq\ldots\subseteq\Incl{w,\tau_n}[j]$.
        We prove by induction on $k$ that, for each position $j$
        and each letter   $x\in\Incl{w,\tau_k}[j]$,
        each inductive 1-set specified by a powerword 
        containing $x$ at position $j$
        contains the word $w$.
        The base is obvious as $x\in\Incl{w,\tau_0}[j]$
        means $x=w[j]$.

        To prove the induction step,
        consider some index $k$, position $j$, 
        and letter $x\in\Incl{w,\tau_k}[j]\setminus\Incl{w,\tau_{k-1}}$.
        By definition of $\Incl{w,\tau_k}$, this means that 
        $x=s_k[j]$ 
        (recall that $s_k\leadsto{}t_k$ is the last transition in $\tau_k$,
        and this transition is not present in $\tau_{k-1}$).
        Consider any inductive 1-set $S$ containing the language
        of $x_{j:\ell}$.
        As $s_k\in{}S$ and $S$ is inductive, $t_k\in{}S$ holds.
        Hence, for some $j'$, the 1-set $S$ contains the language 
        of the atomic proposition $t_k[j']_{j':\ell}$.
        By consistency of $\tau$, 
        the letter $t_k[j']$ is in the set $\Incl{w,\tau_{k-1}}[j]$.
        Thus consider $x'=t_k[j']$ and apply the induction
        hypothesis to $x'$ at position $j'$.
        We obtain $w\in{}S$.

        We have proven that $\Sepf{w}$ denotes a maximal inductive $1$-set
        not containing $w$, i.e. $\Sep{w}$.
\end{proof}
\noindent 
We construct a transducer over the alphabet $\Sigma \times 2^\Sigma$ that transduces a configuration $w$ into the formula $\Sepf{w}$ of a table $\tau$ consistent with and complete for $w$. For this, we need the consistency and completeness summaries of a table.

The basic motivation is: locally, there is only a polynomial-size list of options
how a transition can look like, and we can always reuse a fragment of a transition
with a different continuation as long as transducer states match. 
So what we most care about is which transitions of the RTS-transition transducer
become available at each position, and in which order.

To certify completeness, 
we need to check a claim about all transitions accepted by the transition transducer; 
we apply the subset construction and record its execution.
We keep track of two sets of states: reachable at all, and reachable while 
also adding at least one new excluded letter.

\begin{defi}
\label{def:summaries}
Let $\tau = s_1 \leadsto t_1, \ldots, s_n \leadsto t_n$ be a separation table for a configuration $w$. 
The \emph{consistency summary} is the result of applying the following procedure to $\tau$:
\begin{itemize}
\item Replace each pair of rows $s_i \leadsto t_i$ by a sequence of states providing an accepting run of $\transitions$ on it.\\
(This produces a table with $n$ rows and $\ell+1$ columns, whose entries are states of $\transitions$.)
\item In each column, keep the first occurrence of each state, removing the rest.\\
(The result is a sequence of columns; the columns of the sequence may have possibly different lengths.)
\end{itemize}
\noindent 
The \emph{completeness summary} is the sequence $(Q_0, Q_0'), (Q_1, Q_1') \ldots (Q_\ell, Q_\ell')$ of pairs of sets
of states of $\transitions$, defined inductively as follows for every $j \in [0,\ell]$: 
\begin{itemize}
        \item $Q_0$ is the set of initial states and $Q_0'$ is empty.
\item $Q_{j+1}$ is the set of states reachable from $Q_j$ by means of letters $[a, b]$ such that $b \in \Incl{w,\tau}$.
\item $Q'_{j+1}$ is the set of states reachable from $Q'_j$ by means of letters $[a, b]$ such that $b \in \Incl{w,\tau}$, or reachable from $Q_j$ by means of letters  $[a, b]$ such that $a \notin \Incl{w,\tau}$ and $b \in \Incl{w,\tau}$.
\end{itemize}
\end{defi}
\noindent 
To get to the consistency summary of our running example, we start with the accepting runs of the transitions in the separation table and get

  \begin{center}
  \begin{tabular}{cccccccc}
    $\triangleright$ &
    $q_{0}$ &
    $q_{2}$ &
    $q_{3}$ &
    $q_{3}$ &
    $q_{3}$ &
    $q_{3}$ &
    $q_{3}$ \\
    $\triangleright$ &
    $q_{0}$ &
    $q_{4}$ &
    $q_{4}$ &
    $q_{4}$ &
    $q_{4}$ &
    $q_{4}$ &
    $q_{5}$ \\
    $\triangleright$ &
    $q_{0}$ &
    $q_{1}$ &
    $q_{1}$ &
    $q_{1}$ &
    $q_{1}$ &
    $q_{2}$ &
    $q_{3}$ \\
    $\triangleright$ &
    $q_{0}$ &
    $q_{1}$ &
    $q_{1}$ &
    $q_{1}$ &
    $q_{2}$ &
    $q_{3}$ &
    $q_{3}$ \\
    $\triangleright$ &
    $q_{0}$ &
    $q_{1}$ &
    $q_{1}$ &
    $q_{2}$ &
    $q_{3}$ &
    $q_{3}$ &
    $q_{3}$ \\
  \end{tabular}
  \end{center}
\noindent 
By removing all but the first occurrence of each state in each column we get

  \begin{center}
  \begin{tabular}{cccccccc}
    $\triangleright$ &
    $q_{0}$ &
    $q_{2}$ &
    $q_{3}$ &
    $q_{3}$ &
    $q_{3}$ &
    $q_{3}$ &
    $q_{3}$ \\
    $\triangleright$ &
     &
    $q_{4}$ &
    $q_{4}$ &
    $q_{4}$ &
    $q_{4}$ &
    $q_{4}$ &
    $q_{5}$ \\
    $\triangleright$ &
     &
    $q_{1}$ &
    $q_{1}$ &
    $q_{1}$ &
    $q_{1}$ &
    $q_{2}$ &
     \\
    $\triangleright$ &
     &
     &
     &
     &
    $q_{2}$ &
     &
     \\
    $\triangleright$ &
     &
     &
     &
    $q_{2}$ &
     &
     &
     \\
  \end{tabular}
  \end{center}
\noindent 
Since the statement that arises from this table is $\emptyset \; \emptyset \; \emptyset \; \emptyset \; \emptyset \; \emptyset$ the completeness summary becomes
  \begin{equation*}
    \tuple{\set{q_0}, \emptyset} \; \tuple{\set{q_1, q_2, q_4}, \emptyset} \; \tuple{Q'', \emptyset} \tuple{Q'', \emptyset} \; \tuple{Q'', \emptyset} \; \tuple{Q'', \emptyset}
  \end{equation*}
  where $Q'' = \set{q_1, q_2, q_3, q_4, q_5}$ immediately.

Observe that the consistency summary  is a sequence $\alpha =\alpha[1] \ldots
\alpha[\ell]$ where $\alpha[i]$ is a sequence of \emph{distinct} states of
$\transitions$, i.e., an element of $Q_T^{n_T}$, and the completeness summary
is a sequence $\beta = \beta[1] \ldots \beta[\ell]$ where $\beta[i]$ is a pair
of sets  of states of $\transitions$, i.e.,  an element of $2^{Q_T} \times
2^{Q_T}$.
We prove:

\begin{prop}
\label{prop:transducer}
There exists a transducer $T_{sep}$ over the alphabet $\Sigma \times 2^\Sigma$ satisfying the following properties:
\begin{itemize}
        \item The states of $T_{sep}$ are elements of $(Q_T\cup\{\square\})^{n_T} \times (2^{Q_T} \times 2^{Q_T})$ where $n_T$ is the number of states of $\transitions$.
\item There is a polynomial time algorithm that, given two states $q, q'$ of $T_{sep}$ and a letter $\tuple{a,X} \in \Sigma \times 2^\Sigma$ decides whether the triple $(q, \tuple{a,X}, q')$ is a transition of $T_{sep}$.
\item $T_{sep}$ recognizes a word $\tuple{w, \varphi}$ over $\Sigma \times 2^\Sigma$ if{}f $\varphi = \Sepf{w}$.
\end{itemize}
\end{prop}

\begin{proof}
The proof is long. We introduce \emph{local separation refinements} in \autoref{def:localsep}, and give an algorithm for computing them in \autoref{lem:localsep}. We then show that they can be combined into increasingly longer \emph{fragment separation refinements } (\autoref{def:fragment} and \autoref{lem:fragment}).
Equipped with these results, we then prove the Proposition. 

\begin{defi}
\label{def:localsep}
        A \emph{local separation refinement} 
        for a letter $c \in \Sigma$ is a sequence of transitions $\lambda=(q_1,(a_1,b_1),q_1')$,$\ldots$,$(q_n,(a_n,b_n),q_n')$ of $\transitions$
        such that for every $i \in [1, n]$, either $b_i=c$ or $b_i=a_{i'}$ for some $i'<i$. 
  The states $q_1, \ldots, q_n$ and $q_1', \ldots,q_n'$ are called \emph{incoming} and \emph{outgoing} states, respectively. The \emph{first-appearance lists} of $\lambda$ are the result of removing from $q_1, \ldots, q_n$ and $q_1', \ldots,q_n'$ all elements $q_i$ ($q_i'$) such that $q_{k} = q_i$ ($q_k' = q_i'$) for some $k < i$.
\end{defi}
Observe that the condition on $b_i$ corresponds to consistency with the configuration (restricted to a single position).

	\begin{lem}
	\label{lem:localsep}
                There is a polynomial time
                algorithm taking 
                an RTS $\mathcal{R} = \tuple{\Sigma, \initial, \transitions}$,
                a letter $c\in\Sigma$,
                and
                two lists of
                distinct
                states of $\transitions$
                as input.
                The algorithm outputs
                a local separation refinement for $c$ whose first-appearance lists are equal to the input lists.
                The algorithm produces empty output
                if no such refinement exists.
	\end{lem}

\begin{proof}
        The algorithm starts with the empty sequence of transitions, and repeatedly adds transitions until the local separation refinement is constructed. The algorithm keeps track of the transitions, the source letters, and the lists of incoming states and outgoing states of the sequence of transitions constructed so far. We call them the sets of \emph{used} transitions, and source letters, and the lists of \emph{used} incoming states, and outgoing states. Initially all these sets and lists are empty. At each moment in time, a target letter is \emph{permissible} if it is either equal to $c$, or it  has been already used as a source letter. The main loop of the algorithm repeatedly proceeds as follows. First, it checks whether there is a \emph{still unused} transition with a permissible target letter, and whose incoming and outgoing states have already been used (as incoming and outgoing states, respectively).  If so, the transition is added to the current sequence. Otherwise, the algorithm checks if there is an unused transition, whose incoming states and outgoing states have already been used, or are equal to the next states in the input lists. If so, the transition is added to the current sequence. The incoming and the outgoing states are added to the corresponding lists of used states if they are not already present. 
       
If a loop iteration fails to add a new transition, then the loop terminates. If, after termination, the lists of used incoming and outgoing states are equal to the input lists, then the algorithm returns the current sequence of transitions, otherwise it returns nothing.

Each iteration of the algorithm runs in polynomial time. Since each iteration but the last adds one transition to the sequence, the algorithm runs in polynomial time.

We now show that if the algorithm returns a sequence, then that sequence is a local separation refinement for $c$ satisfying the conditions. By construction, after each iteration the current sequence is a local separation refinement; further, the lists of used incoming and outgoing states are first appearances lists, and prefixes of the input lists. Finally, the algorithm returns a sequence only if the lists of used incoming and outgoing states coincide with the input lists. 

It remains to show that if there exists a separation refinement for $c$ satisfying the conditions, then the algorithm returns a sequence.  Assume there is a such a refinement, but the algorithm terminated without returning a sequence. 
Note that if after termination all transitions of the refinement have been used, 
then all incoming and outgoing states of the input lists have been used too, and so the algorithm would have returned a sequence. So at least one transition of the refinement, was not used by the algorithm. Consider the first such transition,
say $(q, (a,b),q')$. Since all the previous transitions of the refinement (and possibly some others) were used, after termination all incoming and outgoing states of the input lists before $q$ and $q'$ have been used too, and $b$ has become  permissible. Then during the last iteration of the main loop the algorithm was able to use $(q, (a,b),q')$, but terminated instead, a contradiction.

This concludes the proof of the lemma.
\end{proof}

\noindent 
We obtain separation refinements for a word $w$ by combining local separation refinements for its letters,
as follows.

\begin{defi}
\label{def:fragment}
        A \emph{run fragment} of an automaton is 
        a run of the automaton obtained by replacing the set of 
        initial states with the set of all states.
        (In other words, the condition of starting in an initial state is dropped.)
        The states starting and ending a run fragment
        are called the \emph{incoming} and \emph{outgoing}
        states of the fragment, respectively.

        A \emph{fragment separation refinement} $\tau$ for a word $w$ of length $\ell$
        is a table of $\ell$ columns $\tau[1]$, $\ldots$,$\tau[\ell]$,
        where each column is a local separation refinement of the same length. 
        Moreover, the sequences of outgoing and incoming states of adjacent columns
        coincide; i.e. for each $j \in[1,\ell]$ the $i$-th outgoing state
        of the $j$-th column is equal to the $i$-th incoming state of the $j+1$-th column. 
        The incoming and outgoing states of $\tau$ are the incoming states of $\tau[1]$ and the 
        outgoing states of $\tau[\ell]$, respectively. 
\end{defi}

        We can also split a fragment separation refinement $\tau$ into rows $\tau_1, \ldots, \tau_n$.
        Each row $\tau_i$ is a run fragment of $\transitions$.
        We prove that fragment separation refinements can be concatenated.

\begin{lem}
\label{lem:fragment}
        Let $\tau_w$ and $\tau_z$ be fragment separation refinements
        for words $w$ and $z$ such that
        the first-appearance list of outgoing states of $\tau_w$ and 
        the first-appearance list of incoming states of $\tau_z$ coincide.  There exists 
        a fragment separation refinement $\tau_{wz}$ for $w z$
        with the same sets of used source letters
        (and therefore the same sets of permissible letters)
        at the corresponding positions.
        Further, $\tau_{wz}$ has the same first-appearance list
        of incoming states as $\tau_w$,
        and the same first-appearances list of outgoing states
        as $\tau_z$.
\end{lem}

\begin{proof}
We construct $\tau_{wz}$ row by row. For this we repeatedly choose
a row of $\tau_w$ and a row of $\tau_z$ such that the outgoing state of $\tau_w$ and the incoming state of $\tau_z$ coincide, concatenate them, and add the result to $\tau_{wz}$. Rows can be chosen multiple times.

We show how to choose the rows of $\tau_w$ and $\tau_z$ so that the final result
is a fragment separation refinement of $\tau_{wz}$ satisfying the conditions of the lemma. It suffices to show how to choose the rows so that they satisfy the following two constraints:
\begin{enumerate}
\item A row of $\tau_w$ or $\tau_z$ can be chosen only if each previous row has been chosen at least once. 
(Note that once a row is chosen once, it can be chosen again without restrictions.)
\item All rows of $\tau_w$ and $\tau_z$ are eventually chosen.
\end{enumerate}
In particular, this ensures that $\tau_{wz}$ has the same first-appearance list
of incoming states as $\tau_w$, and the same first-appearances list of outgoing states
as $\tau_z$. 

        We proceed as follows. At each step, we
        consider all triples $(q,A,B)$ where $A$ is a row of $\tau_w$
        with outgoing state $q$ and $B$ is a row of $\tau_z$ with incoming state $q$.
        (We also formally add a triple $(\infty,\infty,\infty)$ where
        each component has to be used last.)
        We choose any triple such that the addition of $AB$  to $\tau_{wz}$
        respects the order in the first-appearance lists of incoming states of $\tau_w$ and outgoing states of $\tau_z$.
        This guarantees that the resulting table satisfies condition 1. 
        It suffices to show that all triples will be used.

        Assume the contrary. Let $\sigma$ be the common list of first-appearances of outgoing states of $\tau_w$
        and incoming states of $\tau_z$.
        Consider the earliest non-chosen rows $A$ and $B$
        (it is possible that one of them is $\infty$, but if $\infty$ is used,
        then every row has been used in both refinements).
        Let $q$ be the outgoing state of $A$ and $r$ be the incoming state of $B$.
        We have $q \neq r$, because otherwise the procedure can choose $(q=r,A,B)$.
        Without loss of generality, assume that $q$ precedes $r$ in $\sigma$.
        Then $\tau_z$ has a row $C$ before $B$ with incoming state $q$. 
        Since $B$ is the earliest non-chosen row, $C$ can be chosen.
        So the triple $(q,A,C)$ can be chosen, contradicting that $A$ 
        is the earliest non-chosen row. This concludes the proof of the lemma.
\end{proof}

\noindent 
We now proceed to prove \autoref{prop:transducer}. Let us first construct the transducer $T_{sep}$. 
Recall that a state of the transducer is a pair $q = (\alpha, \beta)$ where
$\alpha \in  (Q_T\cup\{\square\})^{n_T}$, $\beta \in (2^{Q_T} \times 2^{Q_T})$, 
and $n_T$ is the number of states of $\transitions$. In every reachable state of $T_{sep}$,
$\alpha$ will be a permutation of a subset of $Q_T$ followed by some number of $\square$ symbols.
We look at these permutations as first-appearance lists, and so we call $\alpha$ the first-appearance list
of $q$. The second component $\beta$ is the completeness summary letter of $q$ (see \autoref{def:summaries}).

A state $q = (\alpha, \beta)$ is initial if{}f

\begin{itemize}
        \item all the states of $Q_T$ appearing in $\alpha$ are initial; and
        \item $\beta = (Q_{T,0},\emptyset)$ where $Q_{T,0}$ 
                are the initial states of $\transitions$.
\end{itemize}

\noindent and final if{}f

\begin{itemize}
        \item all the states of $Q_T$ appearing in $\alpha$ are final; and
        \item the second set of $\beta$ contains no final states.
\end{itemize}
\noindent 
Let us now define the algorithm recognizing the transitions of $T_{sep}$.
Given two states $q=(\alpha, \beta), q'=(\alpha',\beta')$ and a letter $c$, the algorithm either outputs a set of letters $X$, meaning that $(q, \tuple{a, X},q')$ is a transition, or rejects.
The algorithm takes $\alpha, \alpha'$ and $c$, and uses  the algorithm from the \autoref{lem:localsep} to construct
a local separation refinement for $c$, if any exists. If no such refinement exists, the algorithm rejects. Otherwise, the algorithm assigns to $X$ the complement of the set of permissible letters of this refinement, and  checks whether$\beta, \beta'$ satisfy the definition of a completeness summary  (\autoref{def:summaries}) with respect to $X$.
If they do, then the set $X$ is returned. Clearly, the algorithm runs in polynomial time. 

It remains to show that $T_{sep}$ accepts a word $\tuple{w, \varphi}$ over $\Sigma \times 2^\Sigma$ if{}f $\varphi = \Sepf{w}$. 

Assume $\varphi = \Sepf{w}$. 
Let $\tau$ be a separation table for $w$. 
By \autoref{lem:sepchar}, we have $\varphi= \comp{\Incl{w,\tau}[1]} \, \cdots \, \comp{\Incl{w,\tau}[\ell]}$. 
By the definition of $T_{sep}$, the transducer has a run on $\tuple{w, \varphi}$ whose sequence of visited states are the consistency and completeness summaries of $\tau$ for $w$, and the run is accepting.

Assume now that $T_{sep}$ has an accepting run on $\tuple{w, \varphi}$. We have to show that $\varphi = \Sepf{w}$. 
It suffices to show that from the run we can construct 
a separation table for $w$, since all separation tables  produce the same formula defining $\Sep{w}$.
The procedure to construct a separation table from an accepting run goes as follows\footnote{Note that this procedure is an inefficient proof of existence, there is no need to run it as a part of the algorithm.}.
First, the procedure constructs for each transition in the run a local separation refinement, applying \autoref{lem:localsep}. Then it repeatedly applies \autoref{lem:fragment} to produce one fragment separation refinement for
the complete run. Since, by the definition of $T_{sep}$, all states in the initial (resp. final) first-appearance list
are initial (resp. final), the rows of the joint fragment separation refinement are 
accepting runs of $\transitions$, and so transitions of the RTS. Thus we obtain a sequence of transitions, i.e. a table. 
The table is consistent with the word $w$, because this is a local
property ensured by local separation refinements and preserved during 
the fragment merging. To show that the table is complete for $w$,
observe that, in the completeness summary, for each pair $(Q_j, Q_j')$,
$Q_j$ is the set of all the states reachable by reading  a sequence of pairs of letters 
with target letter permissible for their positions, and $Q_j'$ is the subset where at least one of the used source letters
read in the process was not permissible. By the definition of the final states of $T_{sep}$, no states of $Q_\ell'$ are final
states of $\transitions$. Therefore, no transition of the regular transition system
can add a new permissible letter,  which is the completeness condition.
\end{proof}

This concludes the proof of \autoref{prop:transducer}. We now prove our main technical result:

\begin{thm}
\label{thm:autforInd1}
Let $\mathcal{R} = \tuple{\Sigma, \initial, \transitions}$ be an RTS. 
There exists an NFA $\mathcal{A}_1$ over $\Sigma$ satisfying the following properties:
\begin{itemize}
\item The states of $\mathcal{A}_1$ are elements of $(Q_T\cup \{\square\})^{n_T} \times (2^{Q_T} \times 2^{Q_T}) \times \{0,1\} \times Q_I $.
\item There is a polynomial time algorithm that, given two states $q, q'$ of $\mathcal{A}_1$ and a letter $a \in \Sigma$ decides whether the triple $(q, a, q')$ is a transition of $\mathcal{A}_1$.
\item $\lang{\mathcal{A}_1} = \Ind{1}$
\end{itemize}
\end{thm}
\begin{proof}
Let $T_{sep}$ be the transducer over the alphabet $\Sigma \times 2^\Sigma$ of \autoref{prop:transducer}. 
By \autoref{lem:indchar}, $w \in \Ind{1}$ if{}f there exists a $1$-formula $\varphi$ such that $\tuple{w, \varphi} \in
\lang{T_{sep}}$ and $\varphi$ is not an invariant, that is, there exists $u \in \lang{\initial}$ such that $u \not\models \varphi$. 
By \autoref{lem:satisfaction},
there is a deterministic transducer $\overline{\interpretation}$ with two states, say $\{0, 1\}$,  recognizing the pairs $\tuple{\varphi, u}$ such that $u \not\models \varphi$.  
So we get
\begin{align*}
\Ind{1}  & = \{ w \in \Sigma^* \mid \exists \varphi \in (2^\Sigma)^*, u \in \Sigma^* \text{ s.t. } \tuple{w, \varphi} \in  \rel{T_{sep}}, \tuple{\varphi, u} \in \lang{\overline{\interpretation}}, u \in \lang{\initial} \} \\
  & = \proj{( \rel{T_{sep}} \circ \lang{\overline{\interpretation}} \circ \lang{\initial}  )}{1}.
\end{align*}
The automaton  $\mathcal{A}_1$ is obtained from \autoref{prop:transducerops}.
\end{proof}
\noindent 
Observe that a state of $\mathcal{A}_1$ can be stored using space linear in $\initial$ and $\transitions$. This yields:

\begin{cor}
\label{prop:pspace}
Deciding $\Ind{1} \cap \lang{\mathcal{U}} = \emptyset$ is in \textsc{PSPACE}.
\end{cor}
\begin{proof}
Guess a configuration $w$ and an accepting run of $\mathcal{A}_1$ and $\mathcal{U}$ on $w$, step by step. By \autoref{prop:transducer}, this can be done in polynomial space. Apply then \textsc{NPSPACE} = \textsc{PSPACE}.
\end{proof}

\section{Deciding \texorpdfstring{$\Ind{1} \cap \mathcal{L}(\mathcal{U})= \emptyset$}{Ind\_1 ∩ U = ∅} is
\textsc{PSPACE}-hard.}
\label{sec:pspacehardness}

This section presents the proof of the following lemma:

\begin{lem}
  \label{lem:hardness}
 Given an RTS $\rts$ and an NFA $\mathcal{U}$, deciding $\Ind{1} \cap \mathcal{L}(\mathcal{U})= \emptyset$ is
  \textsc{PSPACE}-hard.
\end{lem}
\noindent 
We reduce from the problem of deciding whether a bounded Turing machine of size $n$ that can only use $n$ tape cells accepts when started on the empty tape.
This problem (very similar to the acceptance problem for linearly bounded automata) is known to  be \textsc{PSPACE}-complete.
Given such a machine, we construct a deterministic RTS $\rts$ which, loosely speaking, satisfies the following properties: 1) the execution of $\rts$  from an initial configuration of length $\Theta(t \cdot n)$
simulates the first $t$ steps of the computation of the Turing machine, and 2) $\Ind{1}$ coincides with the set of reachable configurations.

The proof is divided in several parts. 
We first introduce some notations on Turing machines.
Then, we define the RTS $\rts$, first informally and then formally.
Finally, we conduct the reduction.

\newcommand{\bk}{B}

\paragraph{Turing machines.} We fix some notations on Turing machines.   A Turing machine $\mathcal{M}$ consists of
    \begin{itemize}
      \item a set of states $Q$,
      \item an initial state $q_{0} \in Q$  and a final state $q_{f} \in Q$,
      \item an input alphabet $\Sigma$ and a tape alphabet $\Gamma \supsetneq \Sigma$,
      \item a dedicated blank symbol $\bk \in \Gamma \setminus \Sigma$, and
      \item a transition function $\delta \colon Q \times \Gamma \rightarrow Q  \times \Gamma \times \set{L, R}$.
    \end{itemize}
\noindent 
Let $\mathcal{M}$ be any deterministic linearly bounded Turing machine $\mathcal{M}$, meaning that $\mathcal{M}$ only uses $\size{\mathcal{M}}$ tape cells.
We construct an instance $\rts, \mathcal{U}$ of the safety verification problem, of size $O(n)$, such that $\mathcal{M}$ accepts the empty word if and only if $\Ind{1} \cap \lang{\mathcal{U}} \neq \emptyset$.

Let $n$ denote the size  of $\mathcal{M}$ (and so the number of tape cells) plus 1. We represent a configuration of $\mathcal{M}$ as a word $\alpha \in (\Gamma \cup Q)^{n}$ containing exactly one letter in $Q$. That is, by definition $\alpha = \beta \; q \; \eta$ where $\beta \in \Gamma^*$, $q \in Q$, and $\eta \in \Gamma^+$.
If $\alpha$ is the current configuration of  $\mathcal{M}$, then $\mathcal{M}$ is in the state $q$, the content of the tape is $\beta \; \eta$, and the head of $\mathcal{M}$ reads  the first letter of $\eta$. Observe that $\beta \; \eta$ has length $\size{\mathcal{M}}$. The initial configuration $\alpha_0$ is $q_0 \bk^{n-1}$, that is, the tape is initially empty.

Since $\mathcal{M}$ is deterministic,  there is \emph{exactly} one sequence of configurations $
    \alpha_{0} \vdash
    \alpha_{1} \vdash
    \alpha_{2} \vdash
    \ldots
$
  where $\vdash$ is used to describe that $\alpha_{i+1}$ is the successor configuration of $\alpha_{i}$.
  For simplicity, we also allow $\vdash$ to connect two identical configurations if their state is the final $q_f$.
  Consequently, this sequence of configurations is infinite; either because $\mathcal{M}$ loops or because $\mathcal{M}$ ``stutters'' in a final configuration.

 \paragraph{Description of the RTS}
  We construct an RTS that simulates the execution of $\mathcal{M}$ on the empty word.
  The alphabet of the RTS consists of $\Gamma \cup (\Gamma \times  Q)$, and two auxiliary symbols $\bullet$ and $\square$.
  Using these symbols, the execution of the Turing machine can be  encoded as the infinite word
   \begin{equation*}
     \bullet \; \alpha_{0} \; \bullet \;  \alpha_{1} \;  \bullet \; \alpha_{2} \; \bullet \cdots
  \end{equation*}
\noindent The symbol $\bullet$ separates the individual configurations of the Turing machine.
The set of initial configurations of the RTS contains for every $m \geq 0$ the configuration
  \begin{equation*}
     \; \alpha_{0} \; \bullet \; 
   \underbrace{ \;  \square^{n} \; \bullet \;   \;  \square^{n} \; \bullet \;
    \ldots  \bullet \; \square^{n} \;}_{m \text{ times }}.
  \end{equation*}
Intuitively, an initial configuration of the RTS consist of a sequence of $m+1$ ``pages''. Each page has space to ``write'' a configuration of the Turing machine. The first page is already filled with the initial configuration of the Turing machine, the others are still ``empty''. The transitions of the RTS repeatedly replace the $\square$-symbols, from left to right, by the correct letters of the successor configurations $\alpha_1, \alpha_2, \ldots, \alpha_m$ of the Turing machine. That is, repeated application of the transducer eventually reaches the configurations

\begin{align*}
& \bullet \; \alpha_{0} \; \bullet \; \square^{n}   
    \bullet \; \square^{n}   \bullet \ldots \bullet \; \square^{n} \\
& \bullet \; \alpha_{0} \; \bullet \; \alpha_{1} \;
    \bullet \; \square^{n}  \bullet \ldots \bullet \; \square^{n}   \\
& \bullet \; \alpha_{0} \; \bullet \; \alpha_{1} \;
    \bullet \; \alpha_2 \; \bullet \ldots \bullet \; \square^{n}   \\[-0.2cm]
&  \qquad \cdots \\[-0.2cm]
& \bullet \; \alpha_{0} \; \bullet \; \alpha_{1} \;
    \bullet \; \alpha_{2} \; \bullet \ldots \bullet \; \alpha_{m} \;  \;
\end{align*}
\noindent  (with other configurations in between). Note that the symbol at the position $i+n+1$ of a word in this sequence is determined by the symbols at positions $i-1, i, i+1$ of the previous word. In particular, one can define a (partial) function $\Delta \colon \left( \{\bullet\} \cup \Gamma \cup (Q \times \Gamma) \right)^{3}\to  \Gamma \cup (Q \times \Gamma)$ which gives the letter at position $i+n+1$, given the values of the symbols at positions $i-1, i, i+1$. (The function is partial because there are inputs that do not represent a sensible situation; e.g., inputs which contain more than one element of $Q \times \Gamma$.)

The transducer of the RTS non-deterministically guesses a position $i$, stores the elements at positions $i-1, i, i+1$, and counts to the position $i+n+1$ where it changes $\square$ to the value dictated by $\Delta$.\footnote{One can eliminate this non-deterministic guess by marking the position that is $n+1$ steps before the first $\square$ and moving this marker further one step in each transition. The remaining arguments work analogously. We choose to avoid this improvement for readability.} This transducer can be realized with polynomially many states with respect to $\mathcal{M}$. By construction, every reachable configuration of the RTS that does not contain any $\square$-symbol is a prefix of $\bullet \;  \alpha_{0} \;  \bullet \;  \alpha_{1} \;  \bullet \;  \alpha_{2} \;  \bullet \; \alpha_{3} \;  \bullet \;  \ldots$.
Therefore, $\mathcal{M}$  accepts the empty word if{}f the RTS reaches any configuration containing some occurrence of  $q_{f}$ and no occurrence of $\square$.

  \paragraph{Characterizing $\Ind{1}$.}  We prove that a configuration of the RTS with no occurrence of $\square$ is reachable if{}f it belongs to $\Ind{1}$.
  Therefore, $\mathcal{M}$  accepts the empty word if{}f $\Ind{1}$ contains a configuration with no occurrence of $\square$ and some occurrence of $q_{f}$.

  By the definition of the RTS, it suffices to show that every configuration $w \in \Ind{1}$ is of the form
  \begin{equation}
    \label{eq:target}
    \bullet \;  \alpha_{0} \;  \bullet \ldots \bullet \;  \alpha_{i-1} \;  \bullet \;  \beta \square^{k} \;  \bullet \;  \square^{n} \;  \bullet \ldots  \bullet
    \;  \square^{n} \; 
  \end{equation}
  \noindent for some $i \geq 1$ and $k \leq n$ where $\beta$ is the prefix of $\alpha_i$ of length $n-k$.
  Intuitively, these are the configurations reached by the RTS during the process of ``writing down'' the execution of the Turing machine from the corresponding initial configuration of the RTS by ``filling the $\square$s''.

  The proof hinges on the introduction of a few formulas from $\Ind{1}$.
  To this end, we first observe that every transition of the RTS only changes $\square$ symbols to other symbols.
  Therefore, one can immediately see that, for example, the formula that consists of the single atomic proposition $B_{3: \ell}$ is an inductive 1-invariant, and so an element of $\Ind{1}$ for every $\ell$.
  The reason is that every configuration satisfying $B_{3: \ell}$ has $B$ as its third letter, and no transition of the RTS can change it.
  The same applies for all formulas of the form $a_{k:\ell}$ such that $a \neq \square$, and so, in particular, for all  the formulas
  $$\bullet_{1:\ell} \quad  {q_0}_{2:\ell} \quad B_{3:\ell} \quad \ldots \quad B_{n+1:\ell} \quad  \bullet_{n+2:\ell} $$ 
  \noindent So every configuration of $\Ind{1}$ has 
  $\bullet \; \alpha_{0} \; \bullet$ as a prefix (note that $\alpha_0 = q_0 B^{n-1}$).

  Recall the function $\Delta$ defined above, and assume $\Delta(abc) = d$ where $a, b, c, d \in \{\bullet\} \cup \Gamma \cup (Q \times \Gamma)$. We claim that the formula 
  \begin{equation*}
    \varphi = a_{i-1:\ell} \land b_{i:\ell} \land c_{i+1:\ell} \rightarrow (d_{i+n+1:\ell} \lor \square_{i+n+1:\ell})
  \end{equation*}
  is an inductive 1-invariant for all $i>1$. Observe first that $\varphi$ has indeed one single clause (after applying standard equivalences) .
  Further, since the transducer can only change one single $\square$ to a letter different from $\square$, the formula $\varphi$ is also inductive. Finally, an inspection of the initial configurations of the RTS shows that
  they all satisfy $\varphi$, which proves the claim. Let $w$ be any configuration of a length $\ell$ satisfying the left-hand-side of $\varphi$. 
  The only transition that is applicable to $w$ changes the $i+n+1$-th letter from $\square$ to $d_{i+n+1} \colon \ell$. Therefore, in all configurations of 
  $\Ind{1}$ of a length at least $i+n+1$, the letter at position $i+n+1$ is either $\square$ or $d$. It follows that every configuration with a prefix that does not contain $\square$ in $\Ind{1}$ is of the form (\ref{eq:target}).
  Using this and fixing $\mathcal{U}$ to accept configurations without any $\square$ but at least one occurrence of $q_{f}$ ends the proof.

\section{How large must the bound \texorpdfstring{$b$}{b} be?}
\label{sec:experiments}

The index $b$ needed to prove a property (i.e., the least $b$ such that $\Ind{b}$ implies the property)
can be seen as a measure of how difficult it is for
a human to understand the proof. We use the experimental setup of
\cite{DBLP:conf/tacas/BozgaEISW20,DBLP:journals/corr/abs-2108-09101,DBLP:journals/jlap/BozgaIS21,DBLP:conf/apn/EsparzaRW21},
where systems are encoded as WS1S formulas and \texttt{MONA} \cite{mona} is
used as a computation engine, to show that $b=1$ is enough for a substantial
number of benchmarks used in the RMC literature.
Note that our goal is to evaluate the complexity of invariants 
needed for systems from diverse domains, not to present a tool ready to verify complex systems.

\def\sysWidth{4cm}
\def\propWidth{1.5cm}
\def\propDescWidth{6cm}
\def\propCheckWidth{0.5cm}

\begin{figure}[t]
  \caption{Experimental results of using $\Ind{1}$ as abstraction of the
  set of reachable configurations.}
  \label{fig:results}
  \begin{center}
    \scalebox{0.7}{
    \begin{tabular}{m{\sysWidth}ccccc}
            \toprule
      System & $\size{\mathcal{L}_{I}}$ & $\size{\mathcal{L}_{T}}$ & $\size{\Ind{1}}$ & Properties & time (ms) \\
      \midrule
      Bakery &
      $5$ & $9$ & $8$ &
      \begin{tabular}{m{\propDescWidth}m{\propCheckWidth}}
        deadlock & $\checkmark$ \\
        mutual exclusion & $\checkmark$ \\
      \end{tabular}
      & $< 1$ \\
    \midrule
      Burns &
      $5$ & $9$ & $6$ &
      \begin{tabular}{m{\propDescWidth}m{\propCheckWidth}}
        deadlock & $\checkmark$ \\
        mutual exclusion & $\checkmark$ \\
      \end{tabular}
      & $< 1$ \\
    \midrule
      Dijkstra &
      $6$ & $24$ & $22$ &
      \begin{tabular}{m{\propDescWidth}m{\propCheckWidth}}
        deadlock & $\checkmark$ \\
        mutual exclusion & $\checkmark$ \\
      \end{tabular}
      & $1920$ \\
    \midrule
      Dijkstra (ring) &
      $6$ & $17$ & $17$ &
      \begin{tabular}{m{\propDescWidth}m{\propCheckWidth}}
        deadlock & $\checkmark$ \\
        mutual exclusion & $\times$ \\
      \end{tabular}
      & $2$ \\
    \midrule
      D. cryptographers &
      $6$ & $69$ & $11$ &
      \begin{tabular}{m{\propDescWidth}m{\propCheckWidth}}
        one cryptographer paid & $\checkmark$ \\
        no cryptographer paid & $\checkmark$ \\
      \end{tabular}
      & $5$ \\
    \midrule
      Herman, linear &
      $6$ & $7$ & $6$ &
      \begin{tabular}{m{\propDescWidth}m{\propCheckWidth}}
        deadlock & $\times$ \\
        at least one token & $\checkmark$ \\
      \end{tabular}
      & $< 1$ \\
    \midrule
      Herman, ring &
      $6$ & $7$ & $7$ &
      \begin{tabular}{m{\propDescWidth}m{\propCheckWidth}}
        deadlock & $\checkmark$ \\
        at least one token & $\checkmark$ \\
      \end{tabular}
      & $< 1$ \\
    \midrule
      Israeli-Jafon &
      $6$ & $21$ & $7$ &
      \begin{tabular}{m{\propDescWidth}m{\propCheckWidth}}
        deadlock & $\checkmark$ \\
        at least one token & $\checkmark$ \\
      \end{tabular}
      & $< 1$ \\
    \midrule
      Token passing &
      $6$ & $7$ & $7$ &
      \begin{tabular}{m{\propDescWidth}m{\propCheckWidth}}
        at most one token & $\checkmark$ \\
      \end{tabular}
      & $< 1$ \\
    \midrule
      Lehmann-Rabin &
      $5$ & $15$ & $13$ &
      \begin{tabular}{m{\propDescWidth}m{\propCheckWidth}}
        deadlock & $\checkmark$ \\
      \end{tabular}
      & $1$ \\
    \midrule
      LR phils. &
      $6$ & $14$ & $15$ &
      \begin{tabular}{m{\propDescWidth}m{\propCheckWidth}}
        deadlock & $\times$ \\
      \end{tabular}
      & $2$ \\
    \midrule
      LR phils.{\scriptsize (with $\bork_{\ell}$ and $\bork_{r}$)} &
      $5$ & $14$ & $9$ &
      \begin{tabular}{m{\propDescWidth}m{\propCheckWidth}}
        deadlock & $\checkmark$ \\
      \end{tabular}
      & $1$ \\
    \midrule
      Atomic D. phil. &
      $5$ & $12$ & $20$ &
      \begin{tabular}{m{\propDescWidth}m{\propCheckWidth}}
        deadlock & $\checkmark$ \\
      \end{tabular}
      & $5$ \\
    \midrule
      Mux array &
      $6$ & $7$ & $8$ &
      \begin{tabular}{m{\propDescWidth}m{\propCheckWidth}}
        deadlock & $\checkmark$ \\
        mutual exclusion & $\times$ \\
      \end{tabular}
      & $< 1$ \\
    \midrule
      Res. allocator &
      $5$ & $9$ & $8$ &
      \begin{tabular}{m{\propDescWidth}m{\propCheckWidth}}
        deadlock & $\checkmark$ \\
        mutual exclusion & $\times$ \\
      \end{tabular}
      & $< 1$ \\
    \midrule
      Berkeley &
      $5$ & $19$ & $9$ &
      \begin{tabular}{m{\propDescWidth}m{\propCheckWidth}}
        deadlock & $\checkmark$ \\
        custom properties & $\nicefrac{2}{3}$ \\
      \end{tabular}
      & $1$ \\
    \midrule
      Dragon &
      $5$ & $26$ & $11$ &
      \begin{tabular}{m{\propDescWidth}m{\propCheckWidth}}
        deadlock & $\checkmark$ \\
        custom properties & $\nicefrac{6}{7}$ \\
      \end{tabular}
      & $3$ \\
    \midrule
      Firefly &
      $5$ & $18$ & $7$ &
      \begin{tabular}{m{\propDescWidth}m{\propCheckWidth}}
        deadlock & $\checkmark$ \\
        custom properties & $\nicefrac{0}{4}$ \\
      \end{tabular}
      & $1$ \\
    \midrule
      Illinois &
      $5$ & $25$ & $14$ &
      \begin{tabular}{m{\propDescWidth}m{\propCheckWidth}}
        deadlock & $\checkmark$ \\
        custom properties & $\nicefrac{0}{2}$ \\
      \end{tabular}
      & $1$ \\
    \midrule
      MESI &
      $5$ & $13$ & $7$ &
      \begin{tabular}{m{\propDescWidth}m{\propCheckWidth}}
        deadlock & $\checkmark$ \\
        custom properties & $\nicefrac{2}{2}$ \\
      \end{tabular}
      & $< 1$ \\
    \midrule
      MOESI &
      $5$ & $13$ & $10$ &
      \begin{tabular}{m{\propDescWidth}m{\propCheckWidth}}
        deadlock & $\checkmark$ \\
        custom properties & $\nicefrac{7}{7}$ \\
      \end{tabular}
      & $1$ \\
    \midrule
      Synapse &
      $5$ & $16$ & $7$ &
      \begin{tabular}{m{\propDescWidth}m{\propCheckWidth}}
        deadlock & $\checkmark$ \\
        custom properties & $\nicefrac{2}{2}$ \\
      \end{tabular}
      & $1$ \\
      \bottomrule
    \end{tabular}}
  \end{center}
\end{figure}

Our set of benchmarks consists of problems studied in
\cite{DBLP:conf/fmcad/ChenHLR17,AbdullaDHR07,DBLP:conf/tacas/BozgaEISW20,DBLP:journals/corr/abs-2108-09101,DBLP:journals/jlap/BozgaIS21,DBLP:conf/apn/EsparzaRW21}.
In a first step, we use \texttt{MONA} to construct a minimal DFA for $\Ind{1}$.
For this, we write a WS1S formula $\Psi_1(w)$ expressing that, for every
1-formula $\varphi$, if $\varphi$ is an inductive invariant, then $w$ satisfies
$\varphi$. \texttt{MONA} yields a minimal DFA for the configurations $w$
satisfying $\Psi_1$, which is precisely $\Ind{1}$. We then construct the
formula $\Psi_1(w) \wedge \mathit{Unsafe}(w)$, and use MONA to check if it is
satisfiable\footnote{The second formula $\Psi_1(w) \wedge \mathit{Unsafe}(w)$
being unsatisfiable
suffices for verification purposes, but
we use $\Psi_1(w)$ to obtain information on the size of the minimal DFA.}. All
files containing the \texttt{MONA} formulas and the results are provided in
\cite{experiment-files}. The results are shown in \autoref{fig:results}. The
first column gives the name of the example. In the second and third column, we
give the number of states of the minimal DFA for $\mathcal{L}_{I}$ and
$\mathcal{L}_{T}$, which we also compute via \texttt{MONA}. In the next column,
we give the size of the minimal DFA for $\Ind{1}$. The fifth column reports whether
a property is implied by $\Ind{1}$ (indicated by $\checkmark$) or not
(indicated by $\times$). For the cache coherence protocols, we replace
$\checkmark$ with $\nicefrac{k}{m}$ to state that $k$ of $m$ custom safety
properties can be established. 
The last column gives the total running time of \texttt{MONA}\footnote{As
reported by \texttt{MONA}.}.  As we can see, $\Ind{1}$ is strong enough to
satisfy 46 out of 59 properties.

In the  second step, we have studied some of the cases in which $\Ind{1}$ is not
strong enough. A direct computation of the automaton for $\Ind{2}$ from a formula
$\Psi_2(w)$ using \texttt{MONA} fails. (A computation based on the
automata construction of Section~\ref{sec:boundedinv} might yield better
results and will be part of our future work.) Using a combination of the
automatic invariant computation method of
\cite{DBLP:conf/tacas/BozgaEISW20,DBLP:journals/corr/abs-2108-09101,DBLP:journals/jlap/BozgaIS21,DBLP:conf/apn/EsparzaRW21}
and manual inspection of the returned invariants, we can report some results
for some examples.

\subsubsection*{Examples for $\Ind{b}$ with $b > 1$.}

\autoref{fig:results} contains two versions of the dining philosophers in
which philosophers take one fork at a time. All philosophers but one are
right-handed, i.e., take their right fork first, and the remaining philosopher
is left-handed. If the forks ``know'' which philosopher has taken them (i.e.,
if they have states $\bork_{\ell}$ and $\bork_r$ indicating that the left or
the right philosopher  has the fork), then deadlock-freedom can be proved using
$\Ind{1}$. If the states of the forks are just ``free'' and ``busy'', then
proving deadlock-freedom requires $\Ind{3}$, and in fact $\reach = \Ind{3}$
holds. We show how to establish this using the technique of
\cite{DBLP:conf/apn/EsparzaRW21} and some additional reasoning in
Appendix~\ref{sec:lr-dining-philos}.

\label{ex:berkeley}
The Berkeley and Dragon cache coherence protocols are considered as
parameterized systems in \cite{DBLP:journals/fmsd/Delzanno03}. For both examples,
$\Ind{1}$ is too coarse to establish all desired consistency assertions.
In Appendix~\ref{sec:berkeley}, we describe the formalization of both examples
and show that $\Ind{2}$ suffices to obtain the missing assertions.

\section{Conclusion}
\label{sec:conclusions}
We have introduced a regular model checking paradigm that
approaches the set of reachable configurations from above. As already observed in
\cite{AbdullaDHR07,DBLP:conf/fmcad/ChenHLR17}, such an approach does not require widening or acceleration techniques, as is the case when approaching from below. The main novelty with respect to \cite{AbdullaDHR07,DBLP:conf/fmcad/ChenHLR17} is the discovery of a natural sequence of regular invariants converging to the set of
reachable configurations.

Our new paradigm raises several questions. The first one is the exact
computational complexity of checking emptiness of the intersection $\Ind{b}$
and the unsafe configurations. We have shown $\textsc{PSPACE}$-completeness for $b=1$,
and we conjecture that the problem is
already \textsc{EXPSPACE}-complete for all $b \geq 2$. We also think that the
CEGAR techniques used in
\cite{DBLP:conf/apn/EsparzaRW21,DBLP:journals/corr/abs-2108-09101} can be
extended to the RMC setting, allowing one to compute intermediate regular
invariants between $\Ind{b}$ and $\Ind{b+1}$. 
Another interesting research venue
is the combination with acceleration or widening techniques, and the
application of learning algorithms, like the one of
\cite{DBLP:conf/fmcad/ChenHLR17}. Currently these techniques try to compute
some inductive regular invariant, or perhaps one described by small automata,
which may lead to invariants difficult to interpret by humans. A better
approach might be to stratify the search, looking first for invariants for
small values of $b$.

\section*{Acknowledgments}
We are grateful to Ahmed Bouajjani for fruitful discussions.
We thank the anonymous reviewers of this and the earlier versions of the paper for their valuable suggestions on presentation.


\bibliographystyle{alphaurl}
\bibliography{references.bib}

\clearpage

\appendix

  \section{Dining philosophers with one left-handed philosopher}
  \label{sec:lr-dining-philos}
  We sketch the formalization of the case in which the states of the forks are
  only ``free'' and ``busy''. Consider an RTS with $\Sigma = \set{\fork, \bork,
  \tilo, \hilo, \eilo}$. The state $\hilo$ represents philosophers who
  already grabbed the first fork and wait for the second one. All other states
  are used as before. The philosopher at index $1$ takes first the fork at
  index $2$ and then the fork at index $n$, while any other philosopher $i > 1$
  first takes the fork at index $i-1$ and then the fork at index $i+1$ (modulo $n$).
  We modelled right-handed philosophers as taking the lower-index fork first as this looks like we are facing the philosopher.
  Naturally, the opposite choice of notation would not lead to substantial changes. For example, both reversing the direction and single-position cyclic shifts preserve regularity of languages.

  In \cite{DBLP:conf/apn/EsparzaRW21} the absence of deadlocks in this example
  is shown via only a few inductive assertions. These assertions can be
  equivalently expressed as $3$-invariants. Moreover, these assertions are
  actually enough to completely characterize $\reach$ in this example. To this
  end, observe that, analogously to \autoref{ex:phil},
  $\reach$ is completely characterized by the absence of a few \emph{invalid}
  patterns. These patterns separate into three cases: First, a philosopher
  should use some fork, but this fork is still considered free. Second, two
  philosophers are in states that require the same fork. Third, no adjacent
  philosopher currently uses some fork, yet this fork is busy. 
  As long as neither of these violations happens,
  the configuration can be reached from the initial configuration
  with the philosophers picking up the forks they should be using
  in an arbitrary order, absence of the second forbidden pattern 
  guarantees the lack of conflicts. Afterwards, the correct set of forks
  will be in use.

  More formally,
  we get
  \begin{itemize}
    \item $\Sigma \; \left( \Sigma \; \Sigma \right)^{*} \; \fork \;
      \left( \hilo \; | \; \eilo \right) \; \Sigma \;
      \left( \Sigma \; \Sigma \right)^{*}$,
      $\left( \Sigma \; \Sigma \right)^{+} \; \eilo \; \fork
      \; \left( \Sigma \; \Sigma \right)^{*}$,
      $\left( \eilo \; | \; \hilo \right) \; \fork \; \left( \Sigma \;
      \Sigma \right)^{*}$, $\eilo \left( \Sigma \; \Sigma \right)^{*} \;
      \fork$,
    \item $\left( \Sigma \; \Sigma \right)^{+} \; \eilo \; \Sigma \; \left(
      \hilo \; | \; \eilo \right) \; \Sigma \; \left( \Sigma \; \Sigma
      \right)^{*}$,
      $\left( \eilo \; | \; \hilo \right) \; \Sigma \; \left( \eilo \; |
      \; \hilo \right) \; \Sigma \; \left( \Sigma \; \Sigma \right)^{*}$,
      $\eilo \; \Sigma \; \left( \Sigma \; \Sigma \right)^{*} \; \eilo
      \; \Sigma$,
    \item
      $\tilo \; \bork \; \tilo \; \Sigma \; \left( \Sigma \; \Sigma
      \right)^{*}$,
      $\left( \tilo \; | \; \hilo \right) \; \Sigma \; \left( \Sigma \;
      \Sigma \right)^{*} \; \left( \tilo \; | \; \hilo \right) \; \bork$,
      $\left( \Sigma \; \Sigma \right)^{+} \; \left( \tilo \; | \; \hilo
      \right) \; \bork \; \tilo \; \Sigma \; \left( \Sigma \; \Sigma
      \right)^{*}$.
  \end{itemize}
  The absence of these patterns can be established with the following languages
  of inductive $1$-invariants and inductive $3$-invariants:
  \begin{center}
    \scalebox{1.0}{
    \begin{tabular}{cc}
            $\vtuple{
          \set{\eilo}
  }
            \vtuple{
          \emptyset
  }
        \left(
            \vtuple{
            \emptyset
    }
            \vtuple{
            \emptyset
    }
        \right)^{*}
            \vtuple{
          \set{\eilo}
  }
            \vtuple{
          \set{\fork}
  }$
        &
            $\vtuple{
          \set{\tilo, \hilo} \\
          \set{\tilo, \hilo} \\
          \emptyset
  }
            \vtuple{
          \emptyset \\
          \emptyset \\
          \emptyset
  }
        \left(
            \vtuple{
            \emptyset \\
            \emptyset \\
            \emptyset
    }
            \vtuple{
            \emptyset \\
            \emptyset \\
            \emptyset
    }
        \right)^{*}
            \vtuple{
          \emptyset \\
          \set{\tilo, \hilo} \\
          \set{\tilo, \hilo}
  }
            \vtuple{
          \set{\bork} \\
          \emptyset \\
          \set{\bork}
  }$
        \\
        $
            \vtuple{
          \set{\eilo, \hilo}
  }
            \vtuple{
          \set{\fork}
  }
        \vtuple{
          \set{\eilo, \hilo}
  }
        \vtuple{
          \emptyset
  }
        \left(
          \vtuple{
            \emptyset
    }
          \vtuple{
            \emptyset
    }
        \right)^{*}$
        &
        $\vtuple{
          \set{\tilo} \\
          \set{\tilo} \\
          \emptyset
  }
        \vtuple{
          \set{\bork} \\
          \emptyset \\
          \set{\bork}
  }
        \vtuple{
          \emptyset \\
          \set{\tilo} \\
          \set{\tilo}
  }
        \vtuple{
          \emptyset \\
          \emptyset \\
          \emptyset
  }
        \left(
          \vtuple{
            \emptyset \\
            \emptyset \\
            \emptyset
    }
          \vtuple{
            \emptyset \\
            \emptyset \\
            \emptyset
    }
        \right)^{*}$
        \\
        $
        \vtuple{
          \emptyset
  }
        \vtuple{
          \emptyset
  }
        \left(
        \vtuple{
          \emptyset
  }
        \vtuple{
          \emptyset
  }
        \right)^{*}
        \vtuple{
          \set{\eilo}
  }
        \vtuple{
          \set{\fork}
  }
        \vtuple{
          \set{\eilo, \hilo}
  }
        \vtuple{
          \emptyset
  }
        \left(
          \vtuple{
            \emptyset
    }
          \vtuple{
            \emptyset
    }
        \right)^{*}$
        &
        $\vtuple{
          \emptyset \\
          \emptyset \\
          \emptyset
  }
        \vtuple{
          \emptyset \\
          \emptyset \\
          \emptyset
  }
        \left(
          \vtuple{
            \emptyset \\
            \emptyset \\
            \emptyset
    }
          \vtuple{
            \emptyset \\
            \emptyset \\
            \emptyset
    }
        \right)^{*}
        \vtuple{
          \set{\tilo, \hilo} \\
          \set{\tilo, \hilo} \\
          \emptyset
  }
        \vtuple{
          \emptyset \\
          \set{\bork} \\
          \set{\bork}
  }
        \vtuple{
          \set{\tilo} \\
          \emptyset \\
          \set{\tilo} \\
  }
        \vtuple{
          \emptyset \\
          \emptyset \\
          \emptyset
  }
        \left(
        \vtuple{
          \emptyset \\
          \emptyset \\
          \emptyset
  }
        \vtuple{
          \emptyset \\
          \emptyset \\
          \emptyset
  }
        \right)^{*}$
      \end{tabular}
      }
  \end{center}
  Consequently, $\Ind{3}$ and $\reach$ coincide for this example,
  as the invariants forbid everything that is not reachable. 
  This
  immediately implies that $\Ind{3}$ proves deadlock-freedom since the system
  actually is deadlock-free.

  However, $\Ind{2}$ is insufficient to prove deadlock-freedom: assume there
  exists some inductive $2$-invariant $I$ that invalidates that $D = \hilo \;
  \bork \; \tilo \; \fork \; \eilo \; \bork$ can be reached. Then, $I$ must
  separate all elements from $\reach$ and all configurations $D'$ with $D'
  \leadsto^{*} D$ because it is inductive. In particular, $D' = \tilo \; \bork
  \; \eilo \; \bork \; \eilo \; \bork$ and the reachable configuration $\tilo
  \; \bork \; \hilo \; \bork \; \eilo \; \bork$. Hence, one clause of $I$
  contains $\vari{\hilo}{3}{6}$. Consider the following pair of configurations:
  $D'' = \tilo \; \fork \; \hilo \; \fork \; \tilo \; \fork$ and $C = \tilo \;
  \bork \; \hilo \; \fork \; \tilo \; \fork$. $I$ must separate $D''$ from $C$
  since $D'' \leadsto \tilo\; \bork \; \eilo \; \fork \; \tilo \; \fork \leadsto \tilo\; \bork \; \eilo \; \bork \; \hilo \; \fork \leadsto \tilo\; \fork \; \tilo \; \fork \; \hilo \; \fork \leadsto \hilo\; \bork \; \tilo \; \fork \; \hilo \; \fork \leadsto  D$ while $C \in \reach$. Since $D'' \models
  \vari{\hilo}{3}{6}$, this separation is based on the second clause of $I$
  which must contain $\vari{\bork}{2}{6}$. This means $\tilo \; \bork \; \hilo
  \; \fork \; \eilo \; \bork \models I$. Since $I$ is inductive and
  $\tilo \; \bork \; \hilo \; \fork \; \eilo \; \bork \leadsto
  \tilo \; \bork \; \eilo \; \bork \; \eilo \; \bork \leadsto
  \tilo \; \fork \; \tilo \; \fork \; \eilo \; \bork \leadsto
  \hilo \; \bork \; \tilo \; \fork \; \eilo \; \bork = D$, the assumption that
  $I$ exists is wrong. Consequently, $D$ cannot be excluded via inductive
  $2$-invariants.

  \section{\texorpdfstring{$\Ind{2}$}{IndInv₂} for cache coherence protocols Berkeley and Dragon}
  \label{sec:berkeley}
  For both protocols, we follow the specification of
  \cite{DBLP:journals/fmsd/Delzanno03}.

\def\invalid{i}
\def\unowned{u}
\def\exclusive{e}
\def\nonexclusive{s}
  \paragraph*{Berkeley}
  In the Berkeley cache coherence protocol, each cell is in one of four
  different states: invalid ($\invalid$), unowned ($\unowned$), exclusive
  ($\exclusive$), and shared ($\nonexclusive$).
  Initially, all cells are invalid. Consequently, the language of initial
  configurations is $\invalid^{*}$. For the transitions, we consider a few
  different events. The first one is that the memory is read, and the
  corresponding cell does provide some value of it; i.e., the cell is
  \emph{not} in the state $\invalid$. In this case, nothing changes:
\begin{equation*}
  \left(
    \vtuple{\invalid \\ \invalid} \;
    \middle| \; \vtuple{\unowned \\ \unowned} \;
    \middle| \; \vtuple{\exclusive \\ \exclusive} \;
    \middle| \; \vtuple{\nonexclusive \\ \nonexclusive}
  \right)^{*}
  \; \left(
    \vtuple{\unowned \\ \unowned} \;
    \middle| \; \vtuple{\exclusive \\ \exclusive} \;
    \middle| \; \vtuple{\nonexclusive \\ \nonexclusive}
  \right)
  \; \left(
    \vtuple{\invalid \\ \invalid} \;
    \middle| \; \vtuple{\unowned \\ \unowned} \;
    \middle| \; \vtuple{\exclusive \\ \exclusive} \;
    \middle| \; \vtuple{\nonexclusive \\ \nonexclusive}
  \right)^{*}.
\end{equation*}

If, on the other hand, a value is read from some cell that is in the state
$\invalid$, then this memory cell fetches the information without claiming
ownership; i.e., moves into the state $\unowned$. Every other memory cell
observes this process. Thus, cells that previously were in $\exclusive$ move to
$\nonexclusive$ to account for the fact that another memory cell holds the same
information.
\begin{equation*}
  \left(
    \vtuple{\invalid \\ \invalid} \;
    \middle| \; \vtuple{\unowned \\ \unowned} \;
    \middle| \; \vtuple{\exclusive \\ \nonexclusive} \;
    \middle| \; \vtuple{\nonexclusive \\ \nonexclusive}
  \right)^{*}
  \; \left(
    \vtuple{\invalid \\ \unowned} \;
  \right)
  \; \left(
    \vtuple{\invalid \\ \invalid} \;
    \middle| \; \vtuple{\unowned \\ \unowned} \;
    \middle| \; \vtuple{\exclusive \\ \nonexclusive} \;
    \middle| \; \vtuple{\nonexclusive \\ \nonexclusive}
  \right)^{*}.
\end{equation*}

If a value is written to a cell that was invalid before, then this cell claims
exclusive ownership; that is, all other cells are invalidated.
\begin{equation*}
  \left(
    \vtuple{\invalid \\ \invalid} \;
    \middle| \; \vtuple{\unowned \\ \invalid} \;
    \middle| \; \vtuple{\exclusive \\ \invalid} \;
    \middle| \; \vtuple{\nonexclusive \\ \invalid}
  \right)^{*}
  \; \left(
    \vtuple{\invalid \\ \exclusive} \;
  \right)
  \; \left(
    \vtuple{\invalid \\ \invalid} \;
    \middle| \; \vtuple{\unowned \\ \invalid} \;
    \middle| \; \vtuple{\exclusive \\ \invalid} \;
    \middle| \; \vtuple{\nonexclusive \\ \invalid}
  \right)^{*}.
\end{equation*}

If a cell already has exclusive ownership of this information, there is nothing
to be done. If the cell has only shared ownership of the value, all other cells
\emph{that claim shared ownership} are invalidated.
\begin{equation*}
  \left(
    \vtuple{\invalid \\ \invalid} \;
    \middle| \; \vtuple{\unowned \\ \invalid} \;
    \middle| \; \vtuple{\exclusive \\ \exclusive} \;
    \middle| \; \vtuple{\nonexclusive \\ \invalid}
  \right)^{*}
  \; \left(
    \vtuple{\unowned \\ \exclusive} \;
    \middle| \; \vtuple{\nonexclusive \\ \exclusive} \;
  \right)
  \; \left(
    \vtuple{\invalid \\ \invalid} \;
    \middle| \; \vtuple{\unowned \\ \invalid} \;
    \middle| \; \vtuple{\exclusive \\ \exclusive} \;
    \middle| \; \vtuple{\nonexclusive \\ \invalid}
  \right)^{*}.
\end{equation*}

Finally, the cache can decide to drop data at any moment in time. Thus, any
cell might move into the state $\invalid$.
\begin{equation*}
  \left(
    \vtuple{\invalid \\ \invalid} \;
    \middle| \; \vtuple{\unowned \\ \unowned} \;
    \middle| \; \vtuple{\exclusive \\ \exclusive} \;
    \middle| \; \vtuple{\nonexclusive \\ \nonexclusive}
  \right)^{*}
  \; \left(
    \vtuple{\unowned \\ \invalid} \;
    \middle| \; \vtuple{\exclusive \\ \invalid} \;
    \middle| \; \vtuple{\nonexclusive \\ \invalid}
  \right)
  \; \left(
    \vtuple{\invalid \\ \invalid} \;
    \middle| \; \vtuple{\unowned \\ \unowned} \;
    \middle| \; \vtuple{\exclusive \\ \exclusive} \;
    \middle| \; \vtuple{\nonexclusive \\ \nonexclusive}
  \right)^{*}.
\end{equation*}

We pose now the question whether a configuration can be reached, where two
different cells are claiming exclusive access to the same data. The corresponding set
$\mathcal{U}$ corresponds to $\Sigma^{*} \; \exclusive \; \Sigma^{*} \;
\exclusive \; \Sigma^{*}$. As shown in \autoref{fig:results}, $\Ind{1}$ does not prove this
property. Let us see why. Assume there is an inductive
$1$-invariant $I$ which invalidates the bad word $b = \exclusive \;
\exclusive$; that is, $b \not\models I$. Observe now that we can reach $b$ in
one step from $b' = \unowned \; \exclusive$ and $b'' = \exclusive \;
\unowned$. Consequently, $I$ cannot be satisfied by $b'$ or $b''$ either.
Otherwise, since $I$ is inductive, we already get $b \models I$. This means,
$I$ must not contain $\vari{\exclusive}{1}{2}$, $\vari{\exclusive}{2}{2}$,
$\vari{\unowned}{1}{2}$ nor $\vari{\unowned}{2}{2}$. This, however, makes $I$
unsatisfiable for the actually reachable configuration $\unowned \; \unowned$.

Using an adapted version of the semi-automatic approach of
\cite{DBLP:conf/apn/EsparzaRW21} and some additional reasoning led us to the
following language of inductive $2$-invariants which exclude all configurations
from $\mathcal{U}$:
\begin{equation*}
  \vtuple{\emptyset \\ \emptyset}^{*} \;
  \vtuple{\set{\invalid, \nonexclusive, \unowned} \\ \set{\invalid}} \;
  \vtuple{\emptyset \\ \emptyset}^{*} \;
  \vtuple{\set{\invalid} \\ \set{\invalid, \nonexclusive, \unowned}} \;
  \vtuple{\emptyset \\ \emptyset}^{*}.
\end{equation*}
Since $\Ind{2}$ is the strongest inductive $2$-invariant, 
and the provided $2$-invariant excludes all the configurations from $\mathcal{U}$,
$\Ind{2}$ also excludes all the configurations from $\mathcal{U}$.
Thus, $\Ind{2}$ is strong enough to prove the property.

\paragraph*{Dragon}
\def\dirty{\hat{e}}
\def\dirtyshared{\hat{s}}
The Dragon protocol distinguishes five states. As before, we have states for
invalid cells ($\invalid$), cells that maintain an exclusive copy of the data
($\exclusive$), and cells that have a (potentially) shared copy of the data
($\nonexclusive$). In contrast to before, the Dragon protocol does not
invalidate other copies of some data when it is updated. 
Instead
two new states which mirror $\exclusive$ and $\nonexclusive$ are introduced but, additionally,
indicate that the data might have changed. We refer to these states as $\dirty$
and $\dirtyshared$, respectively. Regardless, we initialize all cells as
invalid; i.e., we have the initial language $\invalid^{*}$.

Assume a read from a ``valid'' cell; that is, some cell that is not in the
state $\invalid$. In that case, nothing changes:
\begin{equation*}
  \left(
    \vtuple{\invalid \\ \invalid}
    \middle| \vtuple{\exclusive \\ \exclusive}
    \middle| \vtuple{\nonexclusive \\ \nonexclusive}
    \middle| \vtuple{\dirty \\ \dirty}
    \middle| \vtuple{\dirtyshared \\ \dirtyshared}
  \right)^{*}
  \;
  \left(
    \vtuple{\exclusive \\ \exclusive}
    \middle| \vtuple{\nonexclusive \\ \nonexclusive}
    \middle| \vtuple{\dirty \\ \dirty}
    \middle| \vtuple{\dirtyshared \\ \dirtyshared}
  \right)
  \;
  \left(
    \vtuple{\invalid \\ \invalid}
    \middle| \vtuple{\exclusive \\ \exclusive}
    \middle| \vtuple{\nonexclusive \\ \nonexclusive}
    \middle| \vtuple{\dirty \\ \dirty}
    \middle| \vtuple{\dirtyshared \\ \dirtyshared}
  \right)^{*}.
\end{equation*}

If a read occurs from an invalid cell -- while all cells are invalid -- the
accessed cell becomes an exclusive reference:
\begin{equation*}
  \vtuple{\invalid \\ \invalid}^{*}
  \;
  \vtuple{\invalid \\ \exclusive}
  \;
  \vtuple{\invalid \\ \invalid}^{*}.
\end{equation*}

If not all cells are invalid but a read occurs for an invalid cell, then this
cell obtains a copy of the data, having now a shared reference to the data. 
Moreover, all exclusive references;
i.e., cells in the states $\exclusive$ or $\dirty$, move to their shared
counterparts ($\nonexclusive$ and $\dirtyshared$, respectively).
\begin{equation*}
  \left(
    \vtuple{\invalid \\ \invalid}
    \middle| \vtuple{\exclusive \\ \nonexclusive}
    \middle| \vtuple{\nonexclusive \\ \nonexclusive}
    \middle| \vtuple{\dirty \\ \dirtyshared}
    \middle| \vtuple{\dirtyshared \\ \dirtyshared}
  \right)^{*}
  \;
  \vtuple{\invalid \\ \nonexclusive}
  \;
  \left(
    \vtuple{\invalid \\ \invalid}
    \middle| \vtuple{\exclusive \\ \nonexclusive}
    \middle| \vtuple{\nonexclusive \\ \nonexclusive}
    \middle| \vtuple{\dirty \\ \dirtyshared}
    \middle| \vtuple{\dirtyshared \\ \dirtyshared}
  \right)^{*}.
\end{equation*}
It is possible here that, although all cells are invalid, the changing cell becomes only shared.
Since this configuration can also be reached from the configuration where there are exactly two cells in state $\nonexclusive$ and all others in $\invalid$ by one cell moving from its shared state to its invalid state, this does not change the set of actually reachable configurations and, consequently, is immaterial for the correctness analysis.

Writing a cell in the state $\dirty$ does not change anything. On the other hand,
writing a cell in the state $\exclusive$ moves that cell into the state $\dirty$:
\begin{equation*}
  \left(
    \vtuple{\invalid \\ \invalid}
    \middle| \vtuple{\exclusive \\ \exclusive}
    \middle| \vtuple{\nonexclusive \\ \nonexclusive}
    \middle| \vtuple{\dirty \\ \dirty}
    \middle| \vtuple{\dirtyshared \\ \dirtyshared}
  \right)^{*}
  \;
  \vtuple{\dirty \\ \dirty}
  \;
  \left(
    \vtuple{\invalid \\ \invalid}
    \middle| \vtuple{\exclusive \\ \exclusive}
    \middle| \vtuple{\nonexclusive \\ \nonexclusive}
    \middle| \vtuple{\dirty \\ \dirty}
    \middle| \vtuple{\dirtyshared \\ \dirtyshared}
  \right)^{*}
\end{equation*}
and
\begin{equation*}
  \left(
    \vtuple{\invalid \\ \invalid}
    \middle| \vtuple{\exclusive \\ \exclusive}
    \middle| \vtuple{\nonexclusive \\ \nonexclusive}
    \middle| \vtuple{\dirty \\ \dirty}
    \middle| \vtuple{\dirtyshared \\ \dirtyshared}
  \right)^{*}
  \;
  \vtuple{\exclusive \\ \dirty}
  \;
  \left(
    \vtuple{\invalid \\ \invalid}
    \middle| \vtuple{\exclusive \\ \exclusive}
    \middle| \vtuple{\nonexclusive \\ \nonexclusive}
    \middle| \vtuple{\dirty \\ \dirty}
    \middle| \vtuple{\dirtyshared \\ \dirtyshared}
  \right)^{*}.
\end{equation*}

A write operation on a cell that is the only one in state $\nonexclusive$ or
$\dirtyshared$ results in a change to $\dirty$. If there are other cells in
either state, one moves to $\dirtyshared$ while all others move to
$\nonexclusive$.
\begin{equation*}
  \left(
    \vtuple{\invalid \\ \invalid}
    \middle| \vtuple{\exclusive \\ \exclusive}
    \middle| \vtuple{\dirty \\ \dirty}
  \right)^{*}
  \;
  \vtuple{\dirtyshared \\ \dirty}
  \;
  \left(
    \vtuple{\invalid \\ \invalid}
    \middle| \vtuple{\exclusive \\ \exclusive}
    \middle| \vtuple{\dirty \\ \dirty}
  \right)^{*},
\end{equation*}
\begin{equation*}
  \left(
    \vtuple{\invalid \\ \invalid}
    \middle| \vtuple{\exclusive \\ \exclusive}
    \middle| \vtuple{\dirty \\ \dirty}
  \right)^{*}
  \;
  \vtuple{\nonexclusive \\ \dirty}
  \;
  \left(
    \vtuple{\invalid \\ \invalid}
    \middle| \vtuple{\exclusive \\ \exclusive}
    \middle| \vtuple{\dirty \\ \dirty}
  \right)^{*}
\end{equation*}
and
{\scriptsize
\begin{equation*}
  \begin{aligned}
      &\left(
      \vtuple{\invalid \\ \invalid}
      \middle| \vtuple{\exclusive \\ \exclusive}
      \middle| \vtuple{\dirty \\ \dirty}
      \middle| \vtuple{\nonexclusive \\ \nonexclusive}
      \middle| \vtuple{\dirtyshared \\ \nonexclusive}
    \right)^{*}
    \;
    \left(
      \vtuple{\nonexclusive \\ \nonexclusive}
      \middle| \vtuple{\dirtyshared \\ \nonexclusive}
    \right)
    \;
    \left(
      \vtuple{\invalid \\ \invalid}
      \middle| \vtuple{\exclusive \\ \exclusive}
      \middle| \vtuple{\dirty \\ \dirty}
      \middle| \vtuple{\nonexclusive \\ \nonexclusive}
      \middle| \vtuple{\dirtyshared \\ \nonexclusive}
    \right)^{*}
    \;
    \left(
      \vtuple{\nonexclusive \\ \dirtyshared}
      \middle| \vtuple{\dirtyshared \\ \dirtyshared}
    \right)
    \;
    \left(
      \vtuple{\invalid \\ \invalid}
      \middle| \vtuple{\exclusive \\ \exclusive}
      \middle| \vtuple{\dirty \\ \dirty}
      \middle| \vtuple{\nonexclusive \\ \nonexclusive}
      \middle| \vtuple{\dirtyshared \\ \nonexclusive}
    \right)^{*} \\
    | &\left(
      \vtuple{\invalid \\ \invalid}
      \middle| \vtuple{\exclusive \\ \exclusive}
      \middle| \vtuple{\dirty \\ \dirty}
      \middle| \vtuple{\nonexclusive \\ \nonexclusive}
      \middle| \vtuple{\dirtyshared \\ \nonexclusive}
    \right)^{*}
    \;
    \left(
      \vtuple{\nonexclusive \\ \dirtyshared}
      \middle| \vtuple{\dirtyshared \\ \dirtyshared}
    \right)
    \;
    \left(
      \vtuple{\invalid \\ \invalid}
      \middle| \vtuple{\exclusive \\ \exclusive}
      \middle| \vtuple{\dirty \\ \dirty}
      \middle| \vtuple{\nonexclusive \\ \nonexclusive}
      \middle| \vtuple{\dirtyshared \\ \nonexclusive}
    \right)^{*}
    \;
    \left(
      \vtuple{\nonexclusive \\ \nonexclusive}
      \middle| \vtuple{\dirtyshared \\ \nonexclusive}
    \right)
    \;
    \left(
      \vtuple{\invalid \\ \invalid}
      \middle| \vtuple{\exclusive \\ \exclusive}
      \middle| \vtuple{\dirty \\ \dirty}
      \middle| \vtuple{\nonexclusive \\ \nonexclusive}
      \middle| \vtuple{\dirtyshared \\ \nonexclusive}
    \right)^{*}. \\
  \end{aligned}
\end{equation*}
}

If a value is written to a previously invalid cell, then either this cell moves
to $\dirty$ (assuming all other cells are $\invalid$ as well), while an
occurrence of another cell with this value causes the written cell to become
$\dirtyshared$ and all other cells to move to the state $\nonexclusive$.
\begin{equation*}
  \vtuple{\invalid \\ \invalid}^{*}
  \;
  \vtuple{\invalid \\ \dirty}
  \;
  \vtuple{\invalid \\ \invalid}^{*}
\end{equation*}
and
{\scriptsize
\begin{equation*}
  \begin{aligned}
    &\left(
      \vtuple{\invalid \\ \invalid}
      \middle| \vtuple{\exclusive \\ \nonexclusive}
      \middle| \vtuple{\dirty \\ \nonexclusive}
      \middle| \vtuple{\nonexclusive \\ \nonexclusive}
      \middle| \vtuple{\dirtyshared \\ \nonexclusive}
    \right)^{*}
    \;
    \vtuple{\invalid \\ \dirtyshared}
    \;
    \left(
      \vtuple{\invalid \\ \invalid}
      \middle| \vtuple{\exclusive \\ \nonexclusive}
      \middle| \vtuple{\dirty \\ \nonexclusive}
      \middle| \vtuple{\nonexclusive \\ \nonexclusive}
      \middle| \vtuple{\dirtyshared \\ \nonexclusive}
    \right)^{*}
    \;
    \left(
      \vtuple{\dirty \\ \nonexclusive}
      \middle| \vtuple{\nonexclusive \\ \nonexclusive}
      \middle| \vtuple{\dirtyshared \\ \nonexclusive}
      \middle| \vtuple{\exclusive \\ \nonexclusive}
    \right)
    \;
    \left(
      \vtuple{\invalid \\ \invalid}
      \middle| \vtuple{\exclusive \\ \nonexclusive}
      \middle| \vtuple{\dirty \\ \nonexclusive}
      \middle| \vtuple{\nonexclusive \\ \nonexclusive}
      \middle| \vtuple{\dirtyshared \\ \nonexclusive}
    \right)^{*} \\
    | &\left(
      \vtuple{\invalid \\ \invalid}
      \middle| \vtuple{\exclusive \\ \nonexclusive}
      \middle| \vtuple{\dirty \\ \nonexclusive}
      \middle| \vtuple{\nonexclusive \\ \nonexclusive}
      \middle| \vtuple{\dirtyshared \\ \nonexclusive}
    \right)^{*}
    \;
    \left(
      \vtuple{\dirty \\ \nonexclusive}
      \middle| \vtuple{\nonexclusive \\ \nonexclusive}
      \middle| \vtuple{\dirtyshared \\ \nonexclusive}
      \middle| \vtuple{\exclusive \\ \nonexclusive}
    \right)
    \;
    \left(
      \vtuple{\invalid \\ \invalid}
      \middle| \vtuple{\exclusive \\ \nonexclusive}
      \middle| \vtuple{\dirty \\ \nonexclusive}
      \middle| \vtuple{\nonexclusive \\ \nonexclusive}
      \middle| \vtuple{\dirtyshared \\ \nonexclusive}
    \right)^{*}
    \;
    \vtuple{\invalid \\ \dirtyshared}
    \;
    \left(
      \vtuple{\invalid \\ \invalid}
      \middle| \vtuple{\exclusive \\ \nonexclusive}
      \middle| \vtuple{\dirty \\ \nonexclusive}
      \middle| \vtuple{\nonexclusive \\ \nonexclusive}
      \middle| \vtuple{\dirtyshared \\ \nonexclusive}
    \right)^{*}. \\
  \end{aligned}
\end{equation*}
}

Finally, any cell might drop its content at any point.
\begin{equation*}
  \left(
    \vtuple{\invalid \\ \invalid}
    \middle| \vtuple{\exclusive \\ \exclusive}
    \middle| \vtuple{\nonexclusive \\ \nonexclusive}
    \middle| \vtuple{\dirty \\ \dirty}
    \middle| \vtuple{\dirtyshared \\ \dirtyshared}
  \right)^{*}
  \;
  \left(
    \vtuple{\exclusive \\ \invalid}
    \middle| \vtuple{\dirty \\ \invalid}
    \middle| \vtuple{\dirtyshared \\ \invalid}
    \middle| \vtuple{\nonexclusive \\ \invalid}
  \right)
  \;
  \left(
    \vtuple{\invalid \\ \invalid}
    \middle| \vtuple{\exclusive \\ \exclusive}
    \middle| \vtuple{\nonexclusive \\ \nonexclusive}
    \middle| \vtuple{\dirty \\ \dirty}
    \middle| \vtuple{\dirtyshared \\ \dirtyshared}
  \right)^{*}.
\end{equation*}

We are interested now to establish that the language $\Sigma^{*} \; \dirty \;
\Sigma^{*} \; \dirty \; \Sigma^{*}$ cannot be reached. The proof that $\Ind{1}$
is insufficient to exclude all configurations of $\Sigma^{*} \; \dirty \;
\Sigma^{*} \; \dirty \; \Sigma^{*}$ is straightforward:
Observe that both $\nonexclusive \; \dirty$ and $\dirty \; \nonexclusive$ can
reach $\dirty \; \dirty$ in one step. In consequence, analogously to the
argument used for the Berkeley protocol, any inductive $1$-invariant cannot
distinguish between the reachable $\nonexclusive \; \nonexclusive$ and the
unreachable $\dirty \; \dirty$.

On the other hand, the language
\begin{equation*}
  \vtuple{\emptyset \\ \emptyset}^{*}
  \;
  \vtuple{\set{\dirtyshared, \invalid, \nonexclusive} \\ \set{\invalid}}
  \;
  \vtuple{\emptyset \\ \emptyset}^{*}
  \;
  \vtuple{\set{\invalid} \\ \set{\dirtyshared, \invalid, \nonexclusive}}
  \;
  \vtuple{\emptyset \\ \emptyset}^{*}
\end{equation*}
of inductive $2$-invariants (which arose, again, from an adapted version of the semi-automatic approach of
\cite{DBLP:conf/apn/EsparzaRW21}) induces an abstraction disjoint from $\Sigma^{*} \;
\dirty \; \Sigma^{*} \; \dirty \; \Sigma^{*}$. Consequently, $\Ind{2}$ does as
well.

\end{document}

%% file: macros.tex
\newcommand{\reachconf}[1]{\mathit{Reach}_{#1}}
\newcommand{\invar}{I}
\newcommand{\vari}[3]{ {#1}_{{#2}:{#3}}}

\newcommand{\comp}[1]{\overline{#1}}

\newcommand{\compl}[1]{\textit{comp}({#1})}

\newcommand{\true}{\mathit{true}}
\newcommand{\false}{\mathit{false}}

\newcommand{\Sep}[1]{\mathit{Sep}_{#1}}
\newcommand{\Sepf}[1]{\varphi_{\Sep{#1}^\tau}}
\newcommand{\Incl}[1]{\mathit{In}(#1)}

\newcommand{\set}[1]{\left\{ #1 \right\}}
\newcommand{\tuple}[1]{\left< #1 \right>}

\newcommand{\size}[1]{\left| #1 \right|}

\DeclareMathOperator{\reach}{\mathit{Reach}}

\tikzset{
  every state/.style = {
    minimum size=3.5,
    inner sep=2 pt
  }
}

\newcommand{\tilo}{t}
\newcommand{\hilo}{h}

\newcommand{\eilo}{e}

\newcommand{\fork}{f}

\newcommand{\bork}{b}

\newcommand{\rts}{\mathcal{R}}
\newcommand{\initial}{\mathcal{I}}
\newcommand{\transitions}{\mathcal{T}}
\newcommand{\interpretation}{\mathcal{V}}
\DeclareMathOperator{\Inductive}{\mathit{Ind}_\interpretation}

\newcommand{\vtuple}[1]{{\text{\tiny $\begin{bmatrix} #1 \end{bmatrix}$}}}

\newcommand{\lang}[1]{L({#1})}
\newcommand{\varlang}[1]{L_{#1}}
\newcommand{\rel}[1]{R({#1})}
\newcommand{\varrel}[1]{R_{#1}}
\renewcommand{\O}{\mathcal{O}}
\newcommand{\proj}[2]{{#1}|_{{#2}}}
\newcommand{\projg}[2]{{#1}\big|_{{#2}}}

\newcommand{\Ind}[1]{\mathit{IndInv}_{#1}}
\newcommand{\reaches}{\stackrel{*}\leadsto}
\newcommand{\reachesb}[1]{\stackrel{*}\leadsto_{#1}}
\newcommand{\prelint}{\reaches_\interpretation}